\documentclass[11pt]{article}
\usepackage{fullpage}
\def\anonymous{0}
\def\comments{0}

\usepackage[english]{babel}
\usepackage{hyperref}
\usepackage{xcolor}
\definecolor{myred}{rgb}{0.64, 0.0, 0.0}
\hypersetup{
  colorlinks   = true, 
  urlcolor     = blue, 
  linkcolor    = blue,  
  citecolor   = myred,    
}
\usepackage{url}

\input{./notation.sty}
\usepackage{algorithm}
\usepackage{algpseudocode}
\usepackage{comment}

\DeclareMathOperator{\mpcC}{mprn}
\DeclareMathOperator{\monoC}{S^+}
\DeclareMathOperator{\Assignment}{\Acal}
\newcommand{\F}{\mathbb{F}}
\newcommand{\neqOneDisj}[1]{\mathrm{DISJ}^{\neq 1}_{#1}} 

\newcommand{\codeInd}{t}
\newcommand{\codeEscProb}{d}

\usepackage[textwidth=2.5cm]{todonotes}

\newcommand{\mnote}[3]{\todo[color=#3!40,size=\footnotesize]{\textbf{#2:} #1}}

\ifnum\comments=0
\newcommand{\anote}[1]{}
\newcommand{\bnote}[1]{}
\newcommand{\pnote}[1]{}
\newcommand{\tnote}[1]{}
\newcommand{\snote}[1]{}
\else
\newcommand{\anote}[1]{\mnote{#1}{A}{white}}
\newcommand{\bnote}[1]{\mnote{#1}{B}{cyan}}
\newcommand{\pnote}[1]{\mnote{#1}{P}{blue}}
\newcommand{\tnote}[1]{\mnote{#1}{T}{green}}
\newcommand{\snote}[1]{\mnote{#1}{S}{gray}}
\fi

\title{Negations are powerful even in small depth}

\ifnum\anonymous=1
\author{}
\date{}
\else
\author{Bruno Cavalar
\thanks{Department of Computer Science, University of Oxford, United Kingdom.
Supported by EPSRC project EP/Z534158/1 on ``Integrated Approach to Computational Complexity: Structure, Self-Reference and Lower Bounds''.
Email: \texttt{bruno.cavalar@cs.ox.ac.uk}}
\and Th\'{e}o Bor\'{e}m Fabris\thanks{Department of Computer Science,
University of Copenhagen, Denmark. Supported by the European Research
Council (ERC) under grant agreement no. 101125652 (ALBA). Email:
\texttt{thfa@di.ku.dk}}
\and Partha Mukhopadhyay\thanks{Chennai Mathematical Institute, India.
Email: \texttt{partham@cmi.ac.in}}
\and Srikanth Srinivasan\thanks{Department of Computer Science, University
of Copenhagen, Denmark. Supported by the European Research Council (ERC)
under grant agreement no. 101125652 (ALBA). Email: \texttt{srsr@di.ku.dk} }
\and Amir Yehudayoff\thanks{Department of Computer Science, The University
of Copenhagen, and The Technion-IIT. Supported by a DNRF Chair grant.
Email: \texttt{amir.yehudayoff@gmail.com}}
}
\fi

\mathcode`l="8000
\begingroup
\makeatletter
\lccode`\~=`\l
\DeclareMathSymbol{\lsb@l}{\mathalpha}{letters}{`l}
\lowercase{\gdef~{\ifnum\the\mathgroup=\m@ne \ell \else \lsb@l \fi}}%
\endgroup

\begin{document}

\maketitle

\vspace{-1em}
\begin{abstract}
We study the power of negation in the Boolean and algebraic settings and show the following results. 
\begin{enumerate}

\item We construct a family of polynomials $P_n$ in $n$ variables,
    all of whose monomials have positive coefficients,
    such that $P_n$ can be
    computed by a depth three circuit of polynomial size but any monotone
    circuit computing it has size $2^{\Omega(n)}$. This is the strongest
    possible separation result between monotone and non-monotone arithmetic
    computations and improves upon all earlier results, including the
    seminal work of Valiant (1980) and more recently by Chattopadhyay,
    Datta, and Mukhopadhyay (2021).  
We then boot-strap this result to prove strong monotone separations for polynomials
of constant degree, which solves an open problem from the survey of Shpilka and Yehudayoff (2010). 

\item  By moving to the Boolean setting, we can prove superpolynomial monotone Boolean circuit lower
bounds for specific Boolean functions, which imply that \emph{all the powers} of certain monotone polynomials
cannot be computed by polynomially sized monotone arithmetic circuits. This leads to a new kind
of monotone vs.\! non-monotone separation in the arithmetic setting.

\item  We then define a collection of problems with
linear-algebraic nature, which are similar to span programs, and prove monotone Boolean circuit lower bounds
for them. In particular, this gives the strongest known monotone lower bounds for functions in uniform
(non-monotone) $\NC^2$. Our construction also leads to an explicit matroid that defines a monotone
function that is difficult to compute, which solves an open problem by Jukna and Seiwert (2020) in the context of the relative powers of greedy and pure dynamic programming algorithms. 
\end{enumerate}

Our monotone arithmetic and Boolean circuit lower bounds are based on known techniques, such as reduction from monotone arithmetic complexity to multipartition communication complexity and
the approximation method for proving lower bounds for monotone Boolean circuits,
but we overcome several new challenges in order to obtain efficient upper bounds using low-depth circuits. 
\end{abstract}

\thispagestyle{empty} 
\newpage

\tableofcontents

\thispagestyle{empty} 
\newpage

\setcounter{page}{1} 
\section{Introduction}
\label{sec:intro}

\subsection{Background and results}

\paragraph{Separation between arithmetic models.}
In his paper ``A single negation is exponentially powerful'', Valiant~\cite{Valiant80} proved that there are $n$-variate monotone polynomials 
in $\VP$ that require a monotone arithmetic circuit of size at least $2^{\Omega(n^{1/2})}$
and that every computation in $\VP$ can be performed with a single negation gate.
Chattopadhyay, Datta and Mukhopadhyay~\cite{CDM21} 
recently strengthened the upper bound to depth-three circuits (also known as a $\Sigma \Pi \Sigma$ circuit); they proved a $2^{\tilde \Omega(n^{1/4})}$ lower bound on the monotone circuit complexity of a monotone polynomial that can be computed by a $\Sigma \Pi \Sigma$ circuit of polynomial size. 
Another recent work shows that the spanning tree polynomials, having $n$ variables and defined over constant-degree expander graphs, have monotone arithmetic circuit complexity $2^{\Omega(n)}$ \cite{CDGM22}. However, the best known non-monotone upper bound for the spanning tree polynomial is a polynomial-size algebraic branching program (ABP) \cite{Moon70} via a determinant computation \cite{Csanky1976, Berkowitz84}.
Our first main result is a strengthening of these results.

\begin{theorem}
\label{thm:main1}
    There is an $n$-variate polynomial $P = P_n$ such that the following hold:
    \begin{enumerate}
        \item $P$ can be computed by a $\Sigma\Pi\Sigma$ circuit of size $O(n^3)$.
        \item Any monotone arithmetic circuit computing $P$ has size at least $2^{\Omega(n)}$.
    \end{enumerate}
\end{theorem}

The family of polynomials $(P_n)_{n \geq 1}$ in Theorem~\ref{thm:main1}
can be computed by a uniform sequence of circuits.
That is, there is a Turing machine that gets $n$ as input, runs in time $\poly(n)$ and outputs a circuit computing $P_n$.
The polynomial $P_n$ is described using an expander graph $G$ with $n/2$ vertices. 
It is of the form
\begin{equation}
\label{eq:intro-Pn}
P_n(x) = \sum_{a\in\bcube{n/2}}\left(\sum_{\{u,v\}\in E(G)}a_ua_v - 1\right)^2\prod_{i=1}^{n/2} x_{i,a_i}.\end{equation}
We prove below that there is a small $\Sigma \Pi \Sigma$ circuit for $P_n$.
Using the fact that $G$ is an expander, we show that the monotone circuit complexity of $P_n$ is high.
This is based on a reduction from monotone arithmetic circuit lower bounds to a problem in communication complexity. Several recent results exploit the connection between monotone circuit lower bounds and communication complexity~\cite{RazYehudayoff11, HY13, Srinivasan19, CDM21, CDGM22, CGM22}. 

All together, this leads to the strongest possible separation between monotone and non-monotone computation in the algebraic setting. 
A single negation gate which performs the subtraction of two monotone depth three circuits is exponentially powerful even compared to general monotone circuits.

A padding argument allows us to deduce new separation between monotone and non-monotone arithmetic circuits for constant degree polynomials. This allows us to solve Open Problem~9 in the survey of Shpilka and Yehudayoff~\cite{SY10}. Before this work, such a separation for polynomials of constant degree, was known only between monotone and general arithmetic \emph{formulas} \cite{HY09}.  
\begin{theorem}
\label{thm:main2}
For every constant $k\in\Naturals$ and $n\in\Naturals$,
there is an $n$-variate polynomial $P = P_{n,k}$ of total degree $k$ such that the following hold:
    \begin{enumerate}
        \item $P$ can be computed by a $\Sigma\Pi\Sigma$ circuit of size $O(n(\log n)^2)$. 
        \item Any monotone arithmetic circuit computing $P$ has size at least $(n/k)^{\Omega(k)}$.
    \end{enumerate}
\end{theorem}

Such a padding argument is only possible because we prove a \emph{strongly} exponential lower bound in Theorem~\ref{thm:main1} (strengthening ~\cite{CDM21}) and because of the simplicity of our upper bounds for the polynomial $P_n$ defined above (it is not clear how to get such a separation from the strongly exponential lower bound of ~\cite{CDGM22}).

\paragraph{Separation between arithmetic and Boolean models.}
One of the major challenges in the study of monotone arithmetic complexity is proving that there are explicit monotone polynomials $P$ such that for every monotone polynomial $Q$, the monotone complexity of $P \cdot Q$ is high~\cite{HY21}. All of our proof techniques seem to fail for this problem.  
We take a step towards this problem.
We identify an explicit monotone polynomial $P$
such that $P$ has small non-monotone arithmetic circuits but every power of $P$ has high monotone complexity. As far as we know, this is the first separation of this kind. 

We do this via a connection to \emph{Boolean} circuit lower bounds. For every $n$-variate monotone polynomial $P$ over $\R$, we can define a monotone Boolean function $f_P:\{0,1\}^n \to \{0,1\}$ via the correspondence 
\[f_P(x) = 1 \iff P(x)>0.\]
Our first observation is that a monotone circuit for $P$ leads to a monotone circuit for $f_P.$

\begin{observation}
If $P$ can be computed by
    a monotone arithmetic circuit of size $s$, then $f_P$ can be computed by a monotone Boolean circuit of size $s$.    
\end{observation}
\begin{proof}[Sketch]
    The Boolean circuit for $s$ is obtained by replacing $+$ gates by $\vee$ gates and $\times$ gates by $\wedge$ gates.
\end{proof}
 One advantage is that $f_P$ corresponds not just to $P$ but also to all the powers of $P$. So, a monotone lower bound for $f_P$ leads to a monotone lower bound to all the powers of $P$. A second important advantage is that we know how to prove lower bounds for monotone Boolean complexity. 
Razborov~\cite{Raz85} and many works that followed~\cite{AlB87, Tardos88, HR00, CKR22, CGRSS25} developed the approximation method for proving Boolean monotone circuit complexity lower bounds.
In particular, the known lower bounds for the perfect bipartite matching
function already imply that any power of the permanent polynomial is
superpolynomially~\cite{Raz85} (and even exponentially~\cite{CGRSS25}) hard
for monotone arithmetic circuits.

We can extend this to show separations. To do this, we also need to upper bound the non-monotone circuit complexity of $P$ or $f_P$. A central tool in proving arithmetic circuit upper bounds in through linear algebra and the determinant~\cite{Csanky1976, Valiant79a, Berkowitz84,MV97}.
All together we get the following strong separation. 

\begin{theorem}
\label{thm:main3}
    There is a monotone $n$-variate polynomial $P = P_n$, over the reals, such that the following hold:
    \begin{enumerate}
        \item $P$ can be computed by an arithmetic circuit of size polynomial in $n$ and depth $O((\log n)^2)$.
        \item Any monotone Boolean circuit computing the Boolean function $f_P$ that corresponds to $P$ has size at least $n^{\Omega((\log n)^{1/2})}$.
    \end{enumerate}
\end{theorem}

It follows that any monotone arithmetic circuit
computing a power of $P$ must be of super-polynomial size as well. 

\paragraph{Separations between Boolean models.} An adaptation of our construction to work over the field $\Fbb_2$ leads to an improved separation between monotone and non-monotone circuits in the Boolean setting as well. 
We take a small detour to explain the context better. The results in
\cite{AlB87, Tardos88} show that there exists a monotone function in $\P$
that requires monotone circuits of size $2^{n^{1/6 - o(1)}}$. Building on
these results, de Rezende and Vinyals~\cite{deRezendeVinyals25} have
recently shown
that there exists a monotone function in $\P$ that requires monotone
circuits of size $2^{n^{1/3 - o(1)}}$. Very recently, Cavalar et
al.~\cite{CGRSS25} improved Razborov's result and showed that the matching
function for bipartite graphs (computable in $\RNC$ \cite{Lovasz79} and
in $\quasiNC$ \cite{FGT19}) 
and another explicit monotone function computable in $\mathbf{L}$ require monotone circuits of size $2^{n^{1/6 - o(1)}}$.
Earlier works~\cite{GKRS18,DBLP:journals/toc/GargGK020} also exhibited an
explicit function in $\NC$
requiring monotone circuits of size
$2^{n^{\eps}}$, where $\eps > 0$ is an unspecified constant.
 
 In this paper, we obtain the following theorem, proving the strongest
 quantitative lower bounds for a function in $\NC$ (except for the $o(1)$ terms, it also matches the best known lower bound for a function in $\P$).

\begin{theorem}
\label{thm:main4}
    For every sufficient large $n\in\Naturals$,
    there is a monotone Boolean function $f:\{0,1\}^n \to \{0,1\}$ such that the following hold:
    \begin{enumerate}
        \item $f$ can be computed by a uniform Boolean circuit of size polynomial in $n$
            and depth $O(\log n)^2$, i.e., $f$ is in uniform $\NC^2$
        \item Any monotone Boolean circuit computing $f$ has size at least $2^{n^{1/3-o(1)}}$.
    \end{enumerate}
\end{theorem}

The function $f$ from Theorem~\ref{thm:main4} has an additional important
property. As any monotone function is specified by its min-terms, it is
interesting to compare the ``complexity'' of the min-terms to the monotone
complexity of the function. Concretely, Jukna and Seiwert~\cite{JS20}
studied monotone functions whose min-terms are the bases of a matroid.
Their motivation was comparing the power of greedy algorithms and pure
dynamic programming algorithms. In their terminology, the function $f$ from
Theorem~\ref{thm:main4} gives an explicit problem for which a greedy
algorithm exactly solves the problem whereas the cost of approximately
solving the problem using pure dynamic programming is high. They proved
that such problems exists using counting arguments, but we provide the
first explicit example of such a function, solving an open problem from
their work (see~\cite[Problem 1]{JS20}). While there have been other lower
bounds recently for monotone Boolean functions defined via linear
algebra~\cite{GKRS18, CGRSS25}, they do not seem to yield a lower bound for
a matroid-based function as above. 

\subsection{Challenges and Proof Ideas}\label{sec:proof-idea}

\paragraph{Arithmetic lower bounds.} In order to prove Theorem~\ref{thm:main1}, we follow an idea of Hrube\v{s} and Yehudayoff~\cite{HY13} in the \emph{non-commutative} setting,\footnote{In this model, variables do not commute.} which was in turn inspired by results from communication complexity~\cite{Razborov92,DBLP:journals/jacm/FioriniMPTW15}.\bnote{Added this citation.} Let $P(x_{1,0},x_{1,1},\ldots, x_{n,0}, x_{n,1})$ be a (set-multilinear) polynomial such that each monomial in $P$ has the form $x_{1,a_1}\cdots x_{n,a_n}$ for some $a\in \{0,1\}^n.$ It is known~\cite{Nisan91} that if $P$ has a `simple'\footnote{The technical term for this is an Algebraic Branching Program (ABP). A similar but more complicated decomposition can be given for general circuits.} monotone non-commutative circuit of size $s$, then we can write
\[
P = \sum_{i=1}^s Q_i\cdot R_i
\]
where each $Q_i$ depends on the first half of the variables ($x_{j,b}$ for $j\in [n/2]$) and $R_i$ on the latter half, and the polynomials are all monotone. This is analogous to the setting of communication complexity, where we can use a small $2$-party communication protocol for a function $f:\{0,1\}^{n/2}\times \{0,1\}^{n/2}\rightarrow \{0,1\}$ to decompose $f$ into a small sum of Boolean (and hence non-negative) \emph{rectangles}. The work of~\cite{HY13} shows how to exploit this connection to obtain a non-monotone versus monotone separation. The separating example is inspired by the \emph{Unique-Disjointness} function. This function has a simple algebraic description via inner products which leads to a non-monotone upper bound, and the lower bound uses ideas from Razborov's lower bound~\cite{Razborov92} on the communication complexity of this problem.  

Our set-up is the standard \emph{commutative} setting, which means that we get a similar decomposition to the one above, except that each $Q_i$ depends on (an unknown subset of) half of the variables (or more precisely $x_{j,b}$ for some $j\in S_i$  where $|S_i| = n/2$). Now, the analogy is no longer with the standard model of communication complexity but with the model of \emph{multipartition} communication complexity, where each rectangle in the decomposition is allowed to use a different partition of the inputs. This model and its variants have been investigated before in works of~\cite{DHJ2004,Hayes2011,HY16}. The work of~\cite{Hayes2011,Jukna15,HY16} in particular discovered a technique to `lift' some structured lower bounds from standard communication complexity to the multipartition setting, by defining variants of the hard functions (for the standard communication model) using an expander graph. In particular, the lower bounds for the Unique-disjointness function are amenable to this lifting technique. The polynomial $P_n$ defined in (\ref{eq:intro-Pn}) is an algebraic variant of this lifted Boolean function (analogous to how the separating example of~\cite{HY13} is related to the Unique-disjointness function). The lower bound argument combines a \emph{non-deterministic} communication complexity lower bound for Unique-Disjointness (this is due to Razborov but we use a version due to~\cite{KaibelWeltge15}). To lift the lower bound, we need to construct a hard distribution that allows to carry out a covering argument. This hard distribution uses some ideas from~\cite{Srinivasan19}.

To get a separation for constant-degree polynomials (hence resolving the open problem from~\cite{SY10}), we use a padding argument. The variables of $P_n$ naturally come in $n$ groups, each of size $2$. We collect these groups into $k$ buckets of size $n/k$ each. There are now $2^{n/k}$ possible monomials in each group and we define a new polynomial $Q_k$ treating these monomials as our variables. It easily follows that $Q_k$ is at least as hard as $P_n$ since a circuit for $Q_k$ easily yields a circuit for $P_n$ by replacing each variable of $Q_k$ by the corresponding monomial. The converse, however, is not clear. Nevertheless, because our depth-$3$ non-monotone circuits for $P_n$ are very simple, we are also able to modify this construction to get a non-monotone circuit for $Q_k$ of small size and depth.

\paragraph{Arithmetic and Boolean upper bounds.}

To prove \Cref{thm:main3,thm:main4},
we consider polynomials and Boolean functions
corresponding to full-rank sets of vectors.
Let $M$ be an $r\times n$ matrix over a field $\Fbb$. We are going to consider both the cases in which the field is a finite field and $\R$.
For any subset $S\subseteq [n]$ of the columns of $M$, let $M[S]$ be the restriction of $M$ to the columns in $S$.  
Let $f_M : \{0,1\}^n \to \{0,1\}$ be the monotone Boolean function associated to $M$ 
defined as
\[\text{$f_M(x)=1$ iff the matrix $M[S]$ has full rank,}\]
where $S=\{i \in [n] : x_i=1\}$ is the set of columns of $M$ that $x$ indicates.
Equivalently, the function $f_M(x)$ evaluates to $1$
iff the linear system $y^{\mathsf{T}} M[S] = 0$
has a unique solution. 
This can be checked by a Boolean circuit of polynomial size and depth $O(\log n)^2$~\cite{ABO99,Mulmuley87}.

When $\Fbb = \Reals$, we also associate with $M$
an $n$-variate multilinear polynomial $P_M$ via
\begin{equation}\label{eqDefinitionHomLinPoly}
    P_M(x)
    \defeq \sum_{\substack{S\subseteq [n]:\\ |S| = r}} \det(M[S])^2\cdot \prod_{i\in S} x_i.
\end{equation}
For $x \in \{0,1\}^n$, it follows that $P_M(x)>0$ iff $f_M(x)=1$. The
Cauchy–Binet formula allows to compute $P_M$ efficiently with a
non-monotone circuit (and in poly-log depth). This leads to our arithmetic
non-monotone upper bounds. 

It is tempting to try to derive our separations with the type of functions considered
in~\cite{GKRS18,CGRSS25}.
These functions correspond to checking whether a certain system of linear
equations is satisfiable.
More concretely,
there is an implicit matrix $M$ over a field $\bbf$,
and the input string is interpreted as encoding
a subset $S$ of the columns of $M$.
The function accepts if 
$M[S]$ spans a fixed target vector $v \neq 0$.
However,
unlike the full-rank property which is tightly related to
the determinant polynomial, it does not seem possible to connect such a
span restriction with an efficiently constructible polynomial.
Moreover, 
whereas 
a full-rank set of vectors is a basis of a vector space (and thus of a matroid),
the same cannot be said of 
vectors spanning~$v$.
This is crucial for our solution of the Problem~1 of~\cite{JS20}.

Finally, 
as we discuss in the \hyperref[sec:open]{\emph{Open Problems}} section,
the function 
$f_M$
can be computed by a rarely studied
monotone model of computation which we name \emph{monotone rank programs}.
These are restrictions of \emph{monotone span programs}, 
a widely studied monotone computational
model~\cite{DBLP:conf/coco/KarchmerW93,DBLP:journals/combinatorica/BabaiGW99}
that computes the functions from~\cite{GKRS18,CGRSS25} and is equivalent to
linear secret-sharing schemes~\cite{DBLP:conf/codcry/Beimel11}.
It is possible that monotone rank programs are much weaker
than monotone span programs,
just like \emph{monotone dependency programs} (another restriction of
monotone span programs)~\cite{DBLP:conf/dimacs/PudlakS96}.
If that is true, then our work 
extends the exponential monotone circuit lower bound
shown for monotone span programs~\cite{GKRS18}
to an even simpler computational model.
\bnote{Should this be in the results section, the technical
discussion, or entirely in the Open Problems section? I fear it may dillute
the punch from the results section. I also fear that it would be missed
in the Open Problems section.}

\paragraph{Monotone Boolean lower bounds.} We now discuss the monotone Boolean circuit lower bounds in this paper. Recall from above that there is a fixed matrix $r\times n$ $M$ and we want to prove a lower bound for the Boolean function $f_M:\{0,1\}^n\rightarrow \{0,1\}$ which decides if a given subset of the columns of $M$ is full-rank or not. 

To prove such a lower bound, we turn to the approximation technique of
Razborov~\cite{Raz85} which has seen many refinements over the
years~\cite{AlB87,Rossman2014, CKR22,CGRSS25}.
Based on these works,
we abstract
out the core of the
approximation method as applied to monotone Boolean circuit lower bounds
into a general and clean statement (Theorem~\ref{thm:sunflower-gen}) below.
Informally, the statement says that to prove a monotone circuit lower bound
for a Boolean function $f$, it suffices to devise two distributions
$\mathcal{D}_0$ and $\mathcal{D}_1$, supported on the $0$s and $1$s of $f$
respectively, that satisfy two conditions:
\begin{itemize}
    \item (Spreadness of $\mathcal{D}_1$): The probability that any $w$ variables are simultaneously set to $1$ under $\mathcal{D}_1$ is at most $q^{-w}$ (for as large a $q$ and $w$ as possible).
    \item (Robust Sunflowers for $\mathcal{D}_0$): Any $w$-DNF of size
        $s^k$ contains a subset that is a \emph{robust sunflower} with
        respect to $\mathcal{D}_0$. Informally, this is a subset of the
        terms that can be 
        replaced by a single term
        without changing the DNF much on inputs
        from $\mathcal{D}_0.$
\end{itemize}
The lower bound obtained from these two conditions is roughly $(q/s)^w.$ At a high-level, the spreadness of $\cald_1$ says that it is supported on strings of low weight, and the sunflower property of $\cald_0$ says that it is supported on strings of high weight. 
By the monotonicity of $f$,
these weight are ``opposite of what they should be''; the 1s of $f$ should be higher than the 0s of $f$.

To use this technique for $f_M,$ we need to define the distributions which in turn depend on the choice of $M.$ The cleanest example is that of a random $r\times n$ matrix over $\F_2$ where $r \approx n^\eps$. Standard calculations show that such a matrix satisfies two linear algebraic properties: 
\begin{itemize}
    \item P1: Most subsets of $r$ columns are linearly independent.
    \item P0: \emph{All} subsets of columns of cardinality $r' = \tilde{\Omega}(r)$ are linearly independent.
\end{itemize}
P1 allows us to easily obtain a spread $\mathcal{D}_1$ with as low a Hamming weight as possible: simply choose a random subset of size $r$ and set these bits to $1$. This is a spread distribution with spreadness parameter $q \approx (n/r).$ 

Choosing the distribution $\mathcal{D}_0$ is slightly trickier: To sample a large number of columns that are linearly \emph{dependent}, we choose a random vector $u\in \F_2^r$ and pick all the columns that are orthogonal to $u.$ This distribution has \emph{uniform} marginals (and thus relatively high Hamming weight) and is also $r'$-wise independent by property P0 above. This last property allows us to prove the robust sunflower property for $\mathcal{D}_0$ by replacing $\mathcal{D}_0$ by the uniform distribution, since low-width $w\approx \sqrt{r'}$ DNFs cannot distinguish between the two~\cite{Bazzi09, Tal17} except with error $\exp(-\tilde{\Omega}(\sqrt{r'}))$. Recent bounds on robust sunflower lemmas with respect to product distributions~\cite{ALWZ21, BellChueluechaWarnke2021,Rao2025} can now be used to argue the desired property with the parameter $s\approx \sqrt{r'}$. 
Optimizing the parameters leads to  choosing $r$ as $r \approx n^{2/3}$
which leads to a lower bound of $\approx \exp(n^{1/3}).$
Interestingly, this is the first 
time Bazzi's theorem is employed 
to show monotone circuit lower bounds.

We note that Bazzi's theorem is the only overhead in our lower bound,
otherwise we could have matched the best known monotone circut lower
bound for an $\mathbf{NP}$ function ($\approx \exp(n^{1/2})$~\cite{CKR22}).
Nonetheless, our reduction to the robust sunflower lemma via Bazzi
is more efficient than the reduction from matching-sunflowers
via ``blocky'' families of~\cite{CGRSS25},
which gives us a lower bound of 
higher order
for an $\NC$ function
(theirs being $\approx \exp(n^{1/6})$).
On the other hand, 
their upper bound is better (the function is in
$\mathbf{L}$),
which makes the results arguably incomparable.

To obtain explicit matrices with the above properties,
we use constructions of binary linear codes with sufficiently high distances and dual distances.
By translating properties P0 and P1 into the coding theory language, we can obtain P0 by proving that the dual distance of the code is at least $r'$, and P1 by proving that its distance is sufficiently large.

Finally, to prove the separation between the arithmetic and Boolean models, we need to argue the above lower bound for a \emph{real} matrix $M$. A natural strategy is to repeat the above argument, say, for random matrices where each column is an independent sign vector $v\in \{-1,1\}^r$. However this leads to a problem with defining the distribution $\mathcal{D}_0$ as above, since the point thus sampled is unlikely to have a high Hamming weight
(by standard anti-concentration arguments~\cite{LittlewoodOfford1943, Erdos1945_LO}, the chance that two random sign vectors are orthogonal is small), and thus $\mathcal{D}_0$ has marginals that are small. 

We need to have many real vectors so that many of their subsets are linearly dependent. {\em How to achieve this?}
The main idea is to choose a sparse matrix. 
We sample each column of $M$ to be a random Boolean vector of low Hamming weight. This now gives a distribution  $\mathcal{D}_0$ that has higher marginals and satisfy a weak form of bounded independence. Unfortunately, we are unable to apply this weak form in conjunction with the strongest robust sunflower lemmas. 
We are, however, able to use it along with the classical sunflower lemma~\cite{ER60} and thus prove a superpolynomial lower bound \emph{\`{a} la} Razborov~\cite{Raz85} (instead of an exponential lower bound).

\paragraph{Hamming balls and subspaces.} On the way to establishing the spreadness property for $\mathcal{D}_1$, we prove a technical result about intersection patterns of subspaces and the Boolean cube, which we believe is independently interesting. Let $\F$ be any field and $V$ be a subspace of $\F^n$ of dimension $d$. It is a standard fact that $|V\cap \{0,1\}^n|\leq 2^d$ (see e.g.~\cite{AY24}). We extend this bound from the full Hamming cube
to a Hamming ball of radius $s$.
We prove that the number of points of the Boolean cube and Hamming weight at most $s$ in $V$ is at most $\binom{d}{\leq s}.$ These bounds are all easily seen to be tight for a vector space $V$ generated by $d$ standard basis vectors.

\subsection{Open problems}
\label{sec:open}

\paragraph{The power of monotone rank programs.}

A \emph{monotone span program of size $k$
over a field $\bbf$} is a $\bbf$-matrix $A$
with $k$ rows,
together with
a labelling of the rows with an input variable from
$\set{x_1,\dots,x_n}$.
The span program accepts a binary string $x \in \blt^n$
if
the set 
$A_x$
of rows whose labels are satisfied by $x$
span the vector $e_1 = (1,0,\dots,0)$
(any other vector can be chosen by a proper change of basis).
There is a linear-size monotone span program computing the functions from~\cite{GKRS18,CGRSS25}.
Using standard tricks, it is not hard to compute $f_M$ with a monotone span
program as well.

Interestingly, the function $f_M$ can also be computed by a seemingly simpler
monotone model which we 
name
\emph{monotone rank programs}.
Just like a monotone span program, a monotone rank program is a matrix
$A$ with a labelling of the rows.
The difference is that a monotone rank program accepts an input $x$ iff $A_x$ is
full-rank.
Such models have been studied in the context of linear secret sharing
schemes~\cite{DBLP:conf/scn/NikovNP04}
and have been called ``non-redundant'' as there are no dependencies
in $A_x$ for any minimally accepting $x$.

A related model called \emph{monotone dependency programs}
accepts $x$ iff $A_x$ is linearly dependent.
This was shown to be exponentially weaker than monotone span programs
in~\cite{DBLP:conf/dimacs/PudlakS96}.
It remains an open problem to determine whether monotone rank programs 
are weaker
than monotone span programs.

\paragraph{Monotone complexity of matching over planar graphs.}

There are two parts for proving a separation in computational complexity. We have two computational classes $A$ and $B$,
and we need to come up with a problem $P$ that belongs to $B$ but does not belongs to $A$.
In our context, $A$ is a monotone class and $B$ is a non-monotone class. 
We have used techniques from communication complexity \cite{Razborov92, KaibelWeltge15} and approximation methods \cite{Raz85} to prove lower bounds against $A$.  
To prove that $P$ belongs to $B$, we have used either the power of depth-3 circuits with negations, or the power of linear algebra. 

There is also a different path for proving separations. In his
work~\cite{Valiant80}, Valiant used the fact the number of perfect
matchings
in a planar graphs is given by the Pfaffian \cite{Kasteleyn1967} so we can
reduce it to computing the determinant polynomial. 
Again, the upper bound comes from linear algebra. 
Valiant proved that (for some planar graphs) the Pfaffian can be computed
by a non-monotone arithmetic circuit of polynomial size; more
precisely, by an algebraic branching program of polynomial size. But we
are interested in separating Boolean monotone circuits and arithmetic
non-monotone circuits. We know that the upper bound holds, but we do not
know if deciding the perfect matching problem is hard for planar graphs for
monotone Boolean circuits. 

\paragraph{Constant-depth separations.}

\Cref{thm:main1} shows a polynomial computable by depth-3 arithmetic
circuits
requiring maximal monotone arithmetic complexity, improving previous
separations of this kind~\cite{CDM21}.
In the Boolean setting,
a line of
works~\cite{Ajtai1987,Okolnishnikova1982,Raz1992,Grigni1992,DBLP:journals/combinatorica/BabaiGW99,Chen2017,GKRS18,Cavalar2023,CGRSS25}
has either directly or indirectly
studied the relative power of constant-depth and monotone circuits,
recently 
discovering
a monotone function 
computable by constant-depth circuits 
which requires
superpolynomial size monotone circuits~\cite{CGRSS25}.
It remains an open question whether there exists a monotone function 
computable by some constant-depth circuit model 
(in fact, even $O(\log n)$-depth)
requiring exponential-size monotone circuits.
Is there a matrix $M$ such that $f_M$ can exhibit this separation?

This type of question is also interesting from the point of view algebraic
separations.
Is there a polynomial $P$ computable by constant-depth or even
logarithmic-depth arithmetic circuits and such that $f_P$
requires superpolynomial or even exponential size monotone Boolean circuits?
This looks plausible in light of several nontrivial algorithms
which have recently been implemented with constant-depth arithmetic
circuits~\cite{DBLP:conf/focs/0003W24}.

\paragraph{Other questions.}

There are also more specific questions about the techniques we used to obtain our results:
\begin{itemize}
    \item Can we improve our lower bounds in Section~\ref{sec:thm1} to $\Omega\left(\binom{n}{k}\right)$ for an $n$-variate polynomial of constant degree $k$?
    \item Can we make the bounds in Section~\ref{sec:thm3}
        (Theorem~\ref{thm:main3}) exponential? This requires proving a
        better robust sunflower lemma for distributions that are only
        weakly independent (in the sense of Lemma~\ref{lem:D0ind-Reals}).
        One approach would be to show that our distribution 
        fools low-width DNFs. 
        Note that recent work has shown limitations
        on what kinds of distributions can fool
        DNFs~\cite{DBLP:conf/coco/AlekseevGGMR0025}.
    \item Can we make the random choice of matrices $M_n$ in Sections~\ref{sec:thm3} explicit?
        In Section~\ref{sec:thm2}, our construction uses explicit constructions of binary linear codes with high distances and dual distances.
    \item Does the lower bound in Section~\ref{sec:thm3} also hold if each column (in $M_n$) is just $O(1)$-sparse?
\end{itemize}

\ifnum\anonymous=0
\subsection{Acknowledgements}

PM, TBF and SS started their collaboration at the \emph{Workshop on Algebraic Complexity Theory (WACT) 2025} in Bochum and would like to thank the organizers of this workshop for their hospitality.
SS would also like to thank Rohit Gurjar and Roshan Raj for pointing him to polynomial identities related to matroids using the Cauchy-Binet formula; Radu Curticapean for a discussion  related to subgraph  polynomials; Duri Janett for references related to lifting theorems; and Susanna de Rezende for clarifications regarding~\cite{deRezendeVinyals25}.
BC acknowledges support of
EPSRC project EP/Z534158/1 on ``Integrated Approach to Computational
Complexity: Structure, Self-Reference and Lower Bounds''.
\fi

\paragraph{Organization of the paper.}
The rest of the paper is organized as follows. The proofs of Theorem \ref{thm:main1} and Theorem \ref{thm:main2} are presented in Section \ref{sec:thm1}. Section \ref{sec:mon-ckt-lb} contains the necessary framework for our lower bounds of Boolean monotone circuits. The proof of Theorem \ref{thm:main3} and Theorem \ref{thm:main4} are given in Section \ref{sec:thm2} and in Section \ref{sec:thm3} respectively.   

\section{Arithmetic separations}

\label{sec:thm1}

The goal of this section is to prove strong separations between monotone and non-monotone algebraic circuits. 

\subsection{Polynomials from graphs and their $\Sigma\Pi\Sigma$ upper bounds}
\label{sectConstructionGraphPolynomials}

Let $G$ be a graph with vertex-set $[n]$.
Denote by $E(G)$ the edge-set of $G$,
and let $e(G) \defeq |E(G)|$.
Let $X$ be the set of $2n$ variables $X \defeq \{x_{1,0},x_{1,1},\ldots, x_{n,0},x_{n,1}\}$.
Let $P_G$ be the nonnegative polynomial defined by
\[P_G \defeq \sum_{a\in\bcube{n}}\left(\sum_{\{u,v\}\in E(G)}a_ua_v - 1\right)^2\prod_{i=1}^n x_{i,a_i}.\]
For every divisor $k$ of $n$, let
$Q_{k, G}$ be the degree-$k$ polynomial in the $k 2^{n/k}$ variables
$\setbar{x_{i, a}}{i\in[k], a\in\bcube{n/k}}$,
defined by
\[Q_{k,G} \defeq \sum_{a\in\bcube{n}}
\left(\sum_{\{u,v\}\in E(G)}a_ua_v - 1\right)^2\prod_{i=1}^{k} x_{i,(a_{(i-1)n/k + 1},\dotsc,a_{in/k})}.\]
Note that $P_G = Q_{n,G}$.
In the other direction,
by applying the substitution
\begin{equation}\label{eqSubstitutionPolyQtoPolyP}
    x_{i, a}\gets \prod_{j\in[n/k]}x_{(i-1)n/k + j,a_j}
\end{equation}
to the polynomial
    $Q_{k,G}$ we obtain the polynomial $P_G$.
We start by proving some upper bounds for the polynomials $Q_{k,G}$.

\begin{lemma}
\label{lemmaSPSUpperBoundForDegreeKPoly}
    For all $G,k$ as above,
    the polynomial $Q \defeq Q_{k, G}$ has a $\Sigma\Pi\Sigma$ formula of size $O(e(G)^2 k2^{n/k})$.
\end{lemma}
\begin{proof}
For every $p\in\Naturals$, let
    \[Q_p
    \defeq \sum_{a\in\bcube{n}} \left(\sum_{\{u,v\}\in E(G)}a_ua_v\right)^p
    \prod_{i=1}^{k} x_{i,(a_{(i-1)n/k + 1},\dotsc,a_{in/k})}.
    \]
We can write $Q$ as   
\begin{align*}
         Q 
        &= \sum_{a\in\bcube{n}}
        \left(\left(\sum_{\{u,v\}\in E(G)}a_ua_v\right)^2 - 2\left(\sum_{\{u,v\} \in E(G)}a_ua_v\right) + 1\right)
        \prod_{i=1}^{k} x_{i,(a_{(i-1)n/k + 1},\dotsc,a_{in/k})} \\
        &= Q_2 - 2Q_1 + Q_0 .
    \end{align*}
First,
        \begin{align*}
        Q_0
        &= \prod_{i=1}^k\left(\sum_{a\in\bcube{n/k}}x_{i, a}\right) .
        \end{align*}
Second,  
        \begin{align*}
        Q_1
        &= \sum_{\{u,v\} \in E(G)} \sum_{a\in\bcube{n}} a_ua_v
        \prod_{i=1}^{k} x_{i,(a_{(i-1)n/k + 1},\dotsc,a_{in/k})} \\
        &= \sum_{\{u,v\} \in E(G)} Q_{1,uv} 
         \end{align*}
         where
          $$Q_{1,uv} \defeq \sum_{a\in\bcube{n}} a_u a_v
        \prod_{i=1}^{k} x_{i,(a_{(i-1)n/k + 1},\dotsc,a_{in/k})}.$$
    For every $u\in [n]$, let $(i_u, j_u)\in[k]\times [n/k]$ be the unique pair such that
    \[u = (i_u - 1)n/k + j_u.\]
   If $i_u\neq i_v$, then
    \begin{align*}
        Q_{1,uv}
        &=
        \left(\sum_{a\in\bcube{n/k}} a_{j_u} x_{i_u,a}\right)\cdot
        \left(\sum_{a\in\bcube{n/k} } a_{j_v} x_{i_v,a}\right)\\
        & \qquad \cdot
        \prod_{i\in[k]\drop\set{i_u, i_v}} \left(\sum_{a\in\bcube{n/k}} x_{i,a}\right).
    \end{align*}
If $i_u = i_v$, then
    \begin{align*}
        Q_{1,uv}
        &=
        \left(\sum_{a\in\bcube{n/k}} a_{j_u}  a_{j_v} x_{i_u,a}\right)\cdot
        \prod_{i\in[k]\drop\set{i_u}} \left(\sum_{a\in\bcube{n/k}} x_{i,a}\right).
    \end{align*}
Third,        \begin{align*}
        Q_2
        &= \sum_{\{u,v\}, \{r,s\} \in E(G)}  Q_{2,uv,rs}
    \end{align*}
where
$$Q_{2,uv,rs}
    \defeq \sum_{a\in\bcube{n}} a_u a_v a_r a_s\prod_{i=1}^{k} x_{i,(a_{(i-1)n/k + 1},\dotsc,a_{in/k})}.$$
Similarly to the cases for $Q_{1,uv}$ above,
we can compute $Q_{2,uv,rs}$ using a $\Pi \Sigma$ monotone formula of size $O(k 2^k)$.
\end{proof}

\begin{corollary}
    \label{cor:SPS-ubd}
For all $G$ as above,
the polynomial $P_G$ has a $\Sigma\Pi\Sigma$ formula of size $O(e(G)^2n)$.
\end{corollary}
\begin{proof}
    This follows from Lemma~\ref{lemmaSPSUpperBoundForDegreeKPoly} using $k$ equals $n$.
\end{proof}

\subsection{Monotone arithmetic complexity and communication complexity}

Our monotone arithmetic circuit lower bound in based on communication complexity. This connection was utilized in several works~\cite{RazYehudayoff11, Jukna15, Srinivasan19}.
We start by defining the \emph{multipartition rectangle number} of a Boolean function and prove (for completeness) how it is related to the monotone arithmetic complexity of a polynomial.

    Let $X$ be a finite set. We say that $R\subseteq \Pcal(X)$ is a \emph{rectangle} 
    if there is a partition $\set{Y, Z}$ of $X$ such that
    \[R = \setbar{S\cup T}{S\in \Ycal, T\in \Zcal}\]
    for some $\Ycal\subseteq\Pcal(Y)$ and $\Zcal\subseteq\Pcal(Z)$.
    We say that a rectangle is \emph{balanced} if $|Y|, |Z|\in [1/3, 2/3)\cdot |X|$.
    
    For every $a\in\bcube{n}$, let 
    \[\supp(a)\defeq \setbar{i\in[n]}{a_i=1}\]
and for $S\subseteq\bcube{n}$, let
    \[\supp(S)\defeq \setbar{\supp(a)}{a\in S}\subseteq \Pcal([n]).\]
    
The \emph{multipartition rectangle number} $\mpcC(f)$ of $f : \{0,1\}^n \to \{0,1\}$  is the minimum number $r$ such that there is a family $\set{R_i}_{i\in [r]}$ of balanced 
rectangles (with respect to $[n]$) such that
    \[\supp(\inv{f}(1)) = \bigcup_{i\in[r]} R_i.\]

The connection between arithmetic complexity and rectangle numbers is established as follows.
    Let $X\defeq \set{x_1,\dotsc, x_m}$ and let $\Xcal\defeq \set{X_1,\dotsc,X_n}$ be a partition of $X$ to $n$ parts. 
    We say that a polynomial $P$ is \emph{$\Xcal$-set-multilinear} if every monomial of $P$ has exactly one element of $X_i$ for every $i\in[n]$.
    We say that a pair $(g, h)$ of polynomials is a \emph{monotone pair} (with respect to $\Xcal$)
    if there is a partition $\set{Y, Z}$ of $X$ such that $\Xcal$ is a refinement of $\{Y,Z\}$
    and $g\in \Reals[Y]$ and $h\in\Reals[Z]$
    and $g$ and $h$ are monotone polynomials.
    We say that a nonnegative pair $(g, h)$ is \emph{balanced} if $|\Xcal_Y|, |\Xcal_Z|\in [1/3, 2/3)\cdot|\Xcal|$
    where for every $W\subseteq X$,
    \[\Xcal_W\defeq\setbar{B\in\Xcal}{B\subseteq W}.\]
    In words, the number of parts of $\Xcal$ in each of $Y,Z$ is balanced. 

The following structural lemma 
was proved in~\cite{RazYehudayoff11} for general (not necessarily set-multilinear) monotone circuits and
in~\cite{Yehudayoff19} for ordered polynomials. The proof for the set-multilinear case below is basically the same.

\begin{lemma}\label{lemmaNonnegPairDecompForMonotoneCircuits}
    If $P$ is a monotone $\Xcal$-set-multilinear polynomial
    computed by a monotone circuit of size $s$
    then there is a family of balanced monotone pairs $\set{(g_i,h_i)}_{i\in[s]}$ such that 
    \[P = \sum_{i\in[s]} g_ih_i,\]
    and each product $g_ih_i$ only contains monomials from $P$.
\end{lemma}
\begin{proof}[Proof sketch]
In a nutshell, the proof uses that monotone circuits for set-multilinear polynomials are syntactically set-multilinear, and the lemma follows by induction on the number of edges
by locating a single product gate that computes a balanced pair. 
\end{proof}

Now, 
let $X$ be the set of $2n$ variables $X \defeq \{x_{1,0},x_{1,1},\ldots, x_{n,0},x_{n,1}\}$
and $\Xcal \defeq \{X_1,\ldots,X_n\}$ be its partition with $X_i \defeq \{x_{i,0},x_{i,1}\}$.
For any $\Xcal$-set-multilinear monomial $m\defeq \prod_{i=1}^n x_{i,a_i}$, let
    \[\Assignment(m)\defeq \setbar{i\in[n]}{a_i = 1}.\]
For a $\Xcal$-set-multilinear polynomial $P$, let 
$$\Assignment(P)\defeq \setbar{\Assignment(m)}{m\in\supp(P)}\subseteq\Pcal([n]).$$

\begin{lemma}\label{lemmaRelationshipBetweenMultipartCommComplAndMonoArithCompl}
    Let $f : \{0,1\}^n \to \{0,1\}$ and
    let $P$ be a monotone $\Xcal$-set-multilinear polynomial.
    If $\supp(\inv{f}(1)) = \Assignment(P)$, then
    \[\monoC(P) \geq \mpcC(f).\]
\end{lemma}
\begin{proof}
    Let $C$ be a monotone circuit of size $s$ computing $P$.
    By Lemma~\ref{lemmaNonnegPairDecompForMonotoneCircuits}, we know that there is a family $\set{(g_i,h_i)}_{i\in[s]}$ of balanced monotone pairs such that
    \[P = \sum_{i\in[s]}g_ih_i.\]
    This decomposition of $P$ implies the following decomposition of its support
    \[\supp(P) = \bigcup_{i\in[s]}\setbar{m_1m_2}{m_1\in\supp(g_i), m_2\in\supp(h_i)}.\]
    Thus, we get the following decomposition 
    \[\Assignment(P) = \bigcup_{i\in[s]}\setbar{\Assignment(m_1)\cup\Assignment(m_2)}{m_1\in\supp(g_i), m_2\in\supp(h_i)},\]
    which, by hypothesis, implies 
    \[\supp(\inv{f}(1)) = \bigcup_{i\in[s]}\setbar{\Assignment(m_1)\cup\Assignment(m_2)}{m_1\in\supp(g_i), m_2\in\supp(h_i)}.\]
    The final observation is that for any balanced monotone pair $(g, h)$,
    the set 
    \[R\defeq \setbar{\Assignment(m_1)\cup\Assignment(m_2)}{m_1\in\supp(g), m_2\in\supp(h)}\]
    is a balanced rectangle of $[n]$. Together with the decomposition above, this concludes the proof of the lemma.

\end{proof}


\subsection{A multipartition communication complexity lower bound}

In this section, we prove the multipartition communication complexity lower bounds, which in turn implies the monotone circuit lower bound. 
The problem we consider is defined over a graph $G$ with vertex-set $[n]$.
For a subset $U$ of the vertices of $G$,
we denote by $G[U]$ the induced graph of $G$ on $U$.
Denote by $f_G: \{0,1\}^n \to \{0,1\}$ the function defined by
    \[f_G(a)=1 \Longleftrightarrow e(G[\supp(a)])\neq 1.\]
Note that the support of the polynomial $P_G$ is  the same as $f_G^{-1}(1)$.

Our lower bounds hold as long as $G$ is an \emph{expander graph}, since this is required by Lemma~9 from~\cite{Srinivasan19}, which we restated as Lemma~\ref{lemmaMatchingSamplingProperties}.
For concreteness, the definition of expander graph used by~\cite{Srinivasan19} is that $G$ is a $d$-regular graph with the second largest eigenvalue of its adjacency matrix at most $d^{0.75}$.
Our main lemma is the following.
\begin{lemma}\label{lemmaLowerBoundMultipartitionCC}
    If $G$ is an expander graph with vertex-set $[n]$, then
    \[\mpcC(f_G) \geq 2^{\Omega(n)}.\]
\end{lemma}

The main ingredient in the lower bound is the following corruption-like lemma.
It follows from a communication complexity perspective of a result by Kaibel and Weltge~\cite{KaibelWeltge15}; see Lemma 5.10 in Roughgarden's lecture notes~\cite{Roughgarden15}.
\begin{lemma}\label{lemmaRazborovUniqueDisjointnessCorruptionLemma}  
    For $l\in\Naturals$, we define the $2l$-variate \emph{different-from-1 disjointness function} $\neqOneDisj{l}$ as
    \[\neqOneDisj{l}(b_1\ldots, b_l,c_1,\ldots, c_l)=1 \Leftrightarrow |\supp((b_1,\dotsc,b_l))\cap \supp((c_1,\dotsc,c_l))|\neq 1.\]
    Let $l\in [n]$ and $g\defeq \neqOneDisj{l}$.
    For every $i\in[l]$, choose independently and uniformly at random an
    element $(b_i, b_{l+i})$ of the set $\set{(0,0), (1,0), (0,1)}$,
    and let $b \defeq (b_1,\dotsc, b_{2l})$ be the corresponding random assignment over $2l$-variables.
    If $R$ is a rectangle of $[2l]$ with partition
    $\set{\set{1,\dotsc, l}, \set{l+1,\dotsc,2l}}$ such that
    $R\subseteq \supp(\inv{g}(1))$, 
    then
    \[
        \Pr_b[b\in R]< 2^{-l/2}.
    \]
\end{lemma}
\begin{proof}
    As the proof of Lemma 5.10 in the notes~\cite{Roughgarden15} works for any unique-disjointness 1-rectangle, this lemma follows from the fact that any different-from-one disjointness 1-rectangle is also a unique-disjointness 1-rectangle.
\end{proof}

Next, we need to define a ``hard distribution'' on $f_G^{-1}(1)$ such that the measure of all balanced $1$-rectangles contained in $f_G^{-1}(1)$ of is small. While the function $f_G$ bears some similarity to the different-from-1 disjointness function, this does not immediately follow from the lemma above, because we do not a priori know the partition we need to deal with. 

Consider the following algorithm, as defined in~\cite{Srinivasan19},
to sample a random matching from $G$:
for $m\in [n]$ given as input,
\begin{enumerate}
    \item Set $M\gets \emptyset$.
    \item For $i=1,2,\ldots,m$, do the following:
    \begin{enumerate}
        \item Remove all vertices from $G$ that are at distance at most $2$ from any vertex in $M$. Let $G_i$ be the resulting graph.
        \item Choose a uniformly random edge $e_i$ from $E(G_i)$, and add it to $M$.
    \end{enumerate}
    \item Output $M$.
\end{enumerate}
Now given a matching $M$, we use the following procedure to define a random input $a\in\set{0,1}^n$ to $f_G$:
\begin{enumerate}
    \item For each $u\in [n]\drop V(M)$, set $a_u = 0$.
    \item For each $\{u,v\}\in M$, choose independently and uniformly at random an
        element $(a_u, a_v)$ from the $\set{(0,0), (1,0), (0,1)}$.
\end{enumerate}

Note that $e(G[\supp(a)]) = 0$ by our choice of $a$, implying that the procedure above defines a probability distribution over $f_G^{-1}(1)$.
Hence, Lemma~\ref{lemmaLowerBoundMultipartitionCC} immediately follows from the claim that for every balanced rectangle $R \subseteq f_G^{-1}(1)$,
\[ \Pr_a[\supp(a)\in R]\leq 2^{-\Omega(n)}.\]
So it remains to prove this inequality.

The useful properties of the random input $a$ rely on the following properties of the random matching $M$,
which were stated in Lemma 9 from~\cite{Srinivasan19}.
\begin{lemma}\label{lemmaMatchingSamplingProperties}
    If $G$ is an expander with vertex-set $[n]$,
    then there is a constant $\alpha > 0$ such that,
    for $m\defeq \lceil \alpha n \rceil$ and
    for $M$ being a random matching sampled by the above algorithm using $m$ as input,
    the following hold:
    \begin{enumerate}
        \item $M$ is an induced matching of cardinality $m$.
        \item For every balanced partition $\set{A, B}$ of $[n]$,
            \[\Pr_M[|M\cap E_G(A, B)| \leq \gamma m]\leq \exp(-\gamma m),\]
            for an absolute constant $\gamma > 0$.
    \end{enumerate}
\end{lemma}
Let $\set{A, B}$ be any valid partition witnessing that $R$ is a balanced rectangle of $[n]$.
By Lemma~\ref{lemmaMatchingSamplingProperties}, we know that $M$ is an induced matching of $G$ and, with probability at least $1-2^{-\Omega(n)}$, we have
\[|M\cap E_G(A,B)|\geq \gamma m \eqdef s'.\]
For a fixed induced matching 
\[E\defeq \set{\{u_1,v_1\},\dotsc, \{u_s,v_s\}}\subseteq E(G)\]
with $s\geq s'$ and $u_i\in A$ and $v_i\in B$ for every $i\in[s]$,
we condition on the event that the random matching $M$ satisfies 
\[M\cap E_G(A,B) = E.\]
We can now define the following random variables:
for every $i\in [s]$, let
\[b_i\defeq a_{u_i}\text{ and }b_{s+i}\defeq a_{v_i},\]
where $a\in\set{0,1}^n$ is the random assignment associated to $M$.
Let 
\[b\defeq (b_1,\dotsc, b_{2s})\]
be the random $2s$-variable Boolean assignment obtained by the definition above.
By the definition of the assignment $a$,
we obtain that each pair $(b_i, b_{s+i})$ is chosen independently and uniformly at random from the set $\set{(0,0), (1,0), (0,1)}$.
Let $g$ be the $2s$-variate different-from-one disjointness function.
We observe that
\[f_G(a) = 1 \iff g(b) = 1\]
for the assignments $a$ and $b$ defined above. 
Also note that $b$ is a random assignment to $g$ that satisfies the assumptions of Lemma~\ref{lemmaRazborovUniqueDisjointnessCorruptionLemma},
and, furthermore, $\supp(a)\in R$ implies that $\supp(b)\in R_E$, where
\[R_E\defeq
\setbar{S\cup T}
{S\subseteq \set{1,\dotsc,s}, T\subseteq \set{s+1,\dotsc,2s},
(\setbar{u_i}{i\in S}\cup\setbar{v_{i-s}}{i\in T})\in R}
.\]
As $R$ is an rectangle with partition $\set{A,B}$ and we have
$\setbar{u_i}{i\in S}\subseteq A$ and
$\setbar{v_{i-s}}{i\in T}\subseteq B$, we get that
$R_E$ is a rectangle of $[2s]$ with partition $\set{\set{1,\dotsc, s}, \set{s+1,\dotsc,2s}}$.
Moreover, we have $R_E\subseteq \supp(\inv{g}(1))$, since
any $S\cup T\in R_E$ satisfies
\[\setbar{u_i}{i\in S}\cup\setbar{v_{i-s}}{i\in T}\in R\subseteq \supp(\inv{f}(1)),\]
which
implies that $g(1_S,1_T) = 1$, where we have identified $S$ and $T$ with their corresponding indicator vectors.

Thus, $R_E$ is a rectangle satisfying the hypothesis of Lemma~\ref{lemmaRazborovUniqueDisjointnessCorruptionLemma}, so we can apply this lemma to obtain that
\[\Pr_a[b\in R_E] \leq 2^{-\Omega(s)} = 2^{-\Omega(n)},\]
which implies that
\begin{align*}
    \Pr_{M,a}[a\in R\ |\ M\cap E_G(A,B) = E]
    &\leq \Pr_{M,a}[b\in R_E\ |\ M\cap E_G(A, B) = E] \\
    &= \Pr_a[b\in R_E]
    \leq 2^{-\Omega(n)}.
\end{align*}
Therefore, we get
\begin{align*}
    \Pr_{M,a}[a\in R] 
    & = \Pr_{M,a}[a\in R, |M\cap E_G(A, B)| < s'] \\
     &\;+ \sum_{E\subseteq E(G) : |E|\geq s'} \Pr_{M,a}[a\in R\ |\ M\cap E_G(A,B) = E]\cdot \Pr_M[M\cap E_G(A,B) = E] \\
   & \leq 2^{-\Omega(n)}.
\end{align*}

\subsection{Monotone arithmetic circuit lower bounds}

Let us now prove our monotone arithmetic circuit lower bounds.
\begin{corollary}[Theorem~\ref{thm:main1}]
    \label{coroMonotoneArithmeticLowerBound}
    If $G$ is an $n$-vertex expander graph,
    then $\monoC(P_G)\geq 2^{\Omega(n)}$.
\end{corollary}
\begin{proof}
    We have that 
    \begin{align*}
        \Assignment(P_G)
        &= \setbar*{\Assignment\left(\prod_{i=1}^n x_{i,a_i}\right)}
        {a\in\bcube{n}, \sum_{\{u,v\} \in E(G)}a_ua_v \neq 1} \\
        &= \setbar{\supp(a)}{a\in\bcube{n}, e(G[\supp(a)])\neq 1} \\
        &= \supp(\inv{f_G}(1)),
    \end{align*}
    thus, by Lemmas~\ref{lemmaRelationshipBetweenMultipartCommComplAndMonoArithCompl} and~\ref{lemmaLowerBoundMultipartitionCC},
    we obtain
    \[\monoC(P_G) \geq \mpcC(f_G) \geq 2^{\Omega(n)}.\]
\end{proof}

\begin{corollary}[Theorem~\ref{thm:main2}]\label{coroMonotoneArithmeticLowerBoundDegreeK}
    Let $k\in\Naturals$ and $n\in\Naturals$.
    Let $G$ be a graph with $nk$ vertices.
    If $G$ is an expander graph,
    then
    \[\monoC(Q_{k,G})\geq (m/k)^{\Omega(k)},\]
    where $m\defeq k2^n$ is the number of variables in the degree-$k$ polynomial $Q_{k, G}$.
\end{corollary}
\begin{proof}
    By a substitution (Equation~\ref{eqSubstitutionPolyQtoPolyP}),
    we can convert any monotone circuit of size $s$ computing $Q_{k,G}$ into a monotone circuit of size $s + O(mn)$ computing $P_G$.
    Thus, by Corollary~\ref{coroMonotoneArithmeticLowerBound}, we have $2^{\Omega(nk)}\leq s + O(mn)$, so
    \[s
    \geq 2^{\Omega(nk)}
    = \left(2^{n}\right)^{\Omega(k)}
    = \left(m/k\right)^{\Omega(k)}.\]
    Therefore, we obtain
    \[\monoC(Q_{k,G})\geq (m/k)^{\Omega(k)}.
        \qedhere
    \]
\end{proof}

\section{A criterion for monotone Boolean circuit lower bounds}\label{sec:mon-ckt-lb}

The goal here is to describe a criterion for proving monotone Boolean circuit lower bounds for a function $f: \{0,1\}^n \to \{0,1\}$.
The criterion is based on the existence of distributions
 $\cald_0$ and $\cald_1$ over $\{0,1\}^n$ satisfying some properties.
 We start with two important definitions. 
\begin{definition}[Sunflower]
    \label{def:abstract-sunflower}
    For $S \subseteq [n]$,
    let $\indset{S}$ denote the function
    \[  \indset{S}(x) \coloneqq \bigwedge_{i\in S} x_i
    \]
    where by convention $\indset{\emptyset} \equiv 1$.

    We say that a family of sets $\cals \sseq \power{[n]}$
    is a \emph{$(\cald,\eps)$-sunflower}
    if $\card{\cals} \geq 2$
    and
    \begin{equation}
        \label{eq:abstract-sunflower}
        \Pr_{x \flws \cald}
        \left[  
            \exists S \in \cals : 
            \indset{S \sm K}(x)
            =
            1
        \right] > 1-\eps,
    \end{equation}
    where $K \defeq \bigcap_{S\in\cals} S$.
    The family $\cals$ is called $\ell$-uniform if $|S|=\ell$ for every $S \in \cals$.
    Let
    $r(\cald, \ell,\eps)$ be the minimum integer $r$ such that
every $\ell$-uniform family of sets of size at least $r^\ell$ 
    contains a $(\cald,\eps)$-sunflower (if no such $r$ exists
    then it is $\infty$).
\end{definition}

\begin{definition}[Spread]
    For an integer $t$ and $q >0$,
    a distribution $\cald$ over $\bset^n$ 
    is \emph{$t$-wise $q$-spread} if,
    for every $A \sseq [n]$ such that $\card{A} \leq t$,
    $$\Pr_{x \flws \cald}[
        \indset{A}(x) = 1
        ] 
        \leq q^{-\card{A}}.$$
\end{definition}

We show that a monotone circuit lower bound for a function $f\ffrom\bcube{n}\fto\bset$
follows immediately from a spreadness bound on a distribution $\cald_1$ of
accepting inputs,
and a sunflower bound on a distribution $\cald_0$ of rejecting inputs.

\begin{theorem}[A lower bound criterion]
    \label{thm:sunflower-gen}
    Let $n\in\Naturals$ and let $f : \{0,1\}^n \to \{0,1\}$ be a monotone function.
    Let $t\in\Naturals$ and $q\in(0,1)$.
    Let $\cald_0, \cald_1$ be distributions over $\{0,1\}^n$ such that
    $\cald_1$ is $t$-wise $q$-spread, and let
    \[\alpha\defeq \min \set*{\Pr_{x \flws \cald_0}[f(x)=0],\; \Pr_{x \flws \cald_1}[f(x)=1]}. \]
    Let $w\in\Naturals$ such that $w\leq t/2$, and
    define
    \[r_w \defeq \max_{\ell \in [2w]} r(\cald_0, \ell, \alpha n^{-3w}).\]
    If $8r_w\leq q \leq r_w n$,
    then any Boolean monotone circuit computing $f$ has size at least
    \begin{equation*}
        \left( 
            \frac{c \alpha q}{r_w} 
        \right)^w
    \end{equation*}
    where $c$ is a universal positive constant. 
\end{theorem}

We prove Theorem~\ref{thm:sunflower-gen} using the approximation method of
Razborov~\cite{Raz85},
generalising the ``tailored sunflower'' approach
of recent works~\cite{CKR22,DBLP:conf/coco/BlasiokM25,CGRSS25}
which began with~\cite{Rossman2014}.
Certain sunflower criterions have previously appeared
in~\cite{Cavalar2020a,DBLP:conf/coco/BlasiokM25}.
Our presentation of the method is similar to~\cite{CGRSS25},
with the difference that we consider arbitrary distributions.
The following is the main lemma.

\begin{lemma}
    \label{lem:sunflower-lowerbound}
    Let $n\in\Naturals$ and let $f : \{0,1\}^n \to \{0,1\}$ be a monotone function.
    Let $t\in\Naturals$ and $q\in(0,1)$.
    Let $\cald_0, \cald_1$ be distributions over $\{0,1\}^n$ such that
    $\cald_1$ is $t$-wise $q$-spread and 
    $\supp(\cald_i) \sseq f^{-1}(i)$ for every $i \in \blt$.
    Let $w\in\Naturals$ such that $w\leq t/2$, and
    define
    \[r_w \defeq \max_{\ell \in [2w]} r(\cald_0, \ell, n^{-3w}).\]
    If $8r_w\leq q \leq r_w n$,
    then any Boolean monotone circuit computing $f$ has size at least
    \begin{equation*}
        \left( 
            \frac{c q}{r_w} 
        \right)^w
    \end{equation*}
    where $c$ is a universal positive constant. 
\end{lemma}

Before proving our main lemma, let us use it to prove Theorem~\ref{thm:sunflower-gen}.
\begin{proof}[Proof of \Cref{thm:sunflower-gen}]
For every $b \in \blt$, we have
    $\Pr_{x \flws \cald_b}[f(x)=b] \geq \alpha > 0$.
    Let $\cald_b^*$ be the distribution $\cald_b$
    conditioned on the event $f(x)=b$.
First, the distribution $\cald_1^*$ is $t$-wise $(\alpha q)$-spread because for every $A \in \binom{[n]}{\leq t}$, 
    \begin{equation*}
        \Pr_{x \flws \cald_1^*}
        \left[  
            \indset{A}(x) = 1
        \right]
        \leq
        \frac{
            1
        }{\alpha}
        \cdot
        \Pr_{x \flws \cald_1}
        \left[  
            \indset{A}(x) = 1
        \right]
        \leq
        (\alpha q)^{-\card{A}}.
    \end{equation*}
Second, every $(\cald_0, \alpha \eps)$-sunflower $\cals$ is a 
    $(\cald_0^*, \eps)$-sunflower because
    \begin{equation*}
        \Pr_{x \flws \cald_0^*}
        \left[  
            \forall S \in \cals 
            \;\;
            \indset{S \sm K}(x)
            =
            0
        \right] 
        \leq
        \frac{
            1
        }{\alpha}
        \cdot
        \Pr_{x \flws \cald_0}
        \left[  
            \forall S \in \cals 
            \;\;
            \indset{S \sm K}(x)
            =
            0
        \right] 
        <
        \frac{\eps \alpha}{\alpha}
        =
        \eps,
    \end{equation*}
    so $r(\cald_0^*, l, \eps)\leq r(\cald_0, l, \alpha\eps)$.
    The result now follows by \Cref{lem:sunflower-lowerbound}.
\end{proof}

For the rest of this section, we prove Lemma~\ref{lem:sunflower-lowerbound}.
Let $\cald_0,\cald_1,q,t,w,\ell$ be as in the assumptions of the lemma.
We denote by $\cald$ the distribution 
$$\cald \defeq (\cald_0+\cald_1)/2.$$
Our goal is to approximate a given monotone circuit by a monotone DNF formula
\[
    \indfml{\cals} \coloneqq \bigvee_{S\in\cals} \indset{S}
\]
for a given family of sets $\cals \sseq \power{[n]}$ satisfying some properties.

We say that $\indfml{\cals}$ is \emph{$r$-small} if,
for every $\ell\in [n]$,
\[\card{\cals \cap \binom{[n]}{\ell}} \leq r^{\ell}.\]
We say that $\indfml{\cals}$ has \emph{width $w$}
if $\card{S} \leq w$ for every $S \in \cals$.
Finally, we say that $\indfml{\cals}$ is a $(w,r)$-DNF
if it is both $r$-small and has width $w$.
The approximation method now proceeds in the following two steps.

\begin{claim}
    \label{claim:dnfs-approximate}
    There is a universal constant $c \in (0,1)$ such that if a monotone circuit of size at most $(c q/r_w)^w$
    computes $f$
    then there is a 
    $(w, r_w)$-DNF $F$ such that
    \[\Pr_{x \sim \cald}[F(x) = f(x)] \geq 
    0.9.\]
\end{claim}
\begin{claim}
    \label{claim:dnfs-are-dumb}
    For every $\delta \in (0,1/2)$ 
    and for every $(w, \delta q)$-DNF $F$,
    we have
    \[\Pr_{x \sim \cald}[F(x) = f(x)] \leq
    1/2+2 \delta.\]
\end{claim}

Before proving the two claims, we prove the main lemma.

\begin{proof}[Proof of \Cref{lem:sunflower-lowerbound}]
    The proof uses the two claims above.
    When $r_w = \infty$, the theorem trivially holds so assume that $r_w < \infty$.
    Assume (towards a contradiction) that $f$ can be computed by a monotone circuit of size $\leq ( c q/ 2 r_w))^w$ where $c >0$ is the constant from 
    \Cref{claim:dnfs-approximate}.
    As $8r_w \leq q$ by assumption,
    Claim~\ref{claim:dnfs-approximate} implies that there is a $(w,q/8)$-DNF such that
    \[\Pr_{x \sim \cald}[F(x) = f(x)] \geq 
    0.9,\]
    which contradicts \Cref{claim:dnfs-are-dumb}.
\end{proof}

\begin{proof}[Proof of~\Cref{claim:dnfs-are-dumb}]
    Let $F \defeq \indfml{S}$ be a $(w, \delta q)$-DNF.
    If $F \equiv \one$ (i.e., $\emptyset \in \cals$),
    the claim is true by definition of $\cald$ and the fact that $\cald_b$ is supported in $\inv{f}(b)$ for every $b\in\set{0,1}$.
    Otherwise, note that $F(x)=1$ only if there exists
    $S \in \cals$ such that $\indset{S}(x)=1$.
    As $F$ is $\delta q$-small and has width $w$ and $\cald_1$ is $t$-wise $q$-spread with $w\leq t/2$,
    we have
    \begin{equation*}
        \Pr_{x \sim \cald_1}[F(x) = 1]
        \leq
        \sum_{k \in[w]}\sum_{S\in\cals\cap\binom{[n]}{k}}\Pr_{x\sim\cald_1}[\indset{S}(x)=1] 
        \leq
        \sum_{k\in[w]}
        (\delta q)^{k}q^{-k}
        \leq 2 \delta.
        \qedhere
    \end{equation*}
\end{proof}

\begin{proof}[Proof of~\Cref{claim:dnfs-approximate}]
    For convenience, 
    let 
    $r \defeq r_w$ 
    \bnote{Previous version was saying $r = r_0$ (?). Note that $r_w$
    increases with $w$.}
    and let $\eps\defeq n^{-3w}$.
    Let $C$ be a monotone circuit of size at most $(c q/r)^w$ computing
    $f$ where $c>0$ is to be determined. 
    We will construct a $(w,r)$-DNF for $C$
    gate-by-gate, inductively, starting at the input gates until we reach the
    output gate.
    Every input variable is already a~$(w,r)$-DNF.
    As we naively
    combine our inductively constructed $(w,r)$-DNFs, the number of
    terms might increase, potentially violating~$r$-smallness.
    In order to maintain smallness of our DNF, we approximate the naive
    combination 
    by the following procedure.
    The subsequent claim
    summarises the properties of the resulting DNF.

    \begin{algorithm}[H]
        \caption{Plucking procedure $\pluck(\cals)$}
        \label{alg:pluck}
        \begin{algorithmic}[1]
            \While{$\exists \ell\in [2w]\colon |\cals \cap
            \binom{[n]}{\ell}| > r^\ell$}
            \State
            Let $\cals' \subseteq \cals \cap \binom{[n]}{\ell}$ 
            be a $(\cald_0,\eps)$-sunflower with core $K$
            \State
            Let
            $\cals \leftarrow \set{K} \cup \set{S \in \cals : K \not\sseq S}$
            \EndWhile
        \end{algorithmic}
    \end{algorithm}

    Note that Line 2 of the algorithm is always possible by the choice of $r$.
        
    \begin{claim}
        \label{claim:pluck}
        If $F_{\cals}$ has width $2w$, then
        $F_{\pluck(\cals)}$ is 
        a $(2w,r)$-DNF with $F_{\pluck(\cals)} \geq F_{\cals}$
        and
        \begin{equation*}
            \Pr_{x \sim \cald_0}[F_{\pluck(\cals)}(x) > F_{\cals}(x)]
            \leq
            n^{-w}.
        \end{equation*}
    \end{claim}
    \begin{proof}
        Since the core $K$ is contained in all sets in $\cals'$ and $|\cals'| > 1$, \anote{``all'' to ``one''.}\tnote{one to all}
        the size of $\cals \cap \binom{[n]}{\ell}$ is necessarily reduced at Line 3,
        and eventually we obtain an $r$-small family.
        Therefore, as $\cals$ has width $2w$,
        the algorithm ends in at most $\sum_{i=0}^{2w}\binom{n}{i} \leq
        n^{2w}$ iterations.
        It follows that
        $F_{\pluck(\cals)}$ is 
        a $(2w,r)$-DNF.

        Because the core belongs to all sets in the sunflower, we get $F_{\pluck(\cals)} \geq F_{\cals}$.
        It remains to bound the error on $\cald_0$
        incurred in Line 3.
        Such an error happens only if
        $\indset{K}(x)=1$ and, 
        for all $S \in \cals$, we have 
        $\indset{S}(x) = 0$.
        Thus, the error of a single iteration (sunflower plucking)
        can be bounded by
        \begin{align*}
            \Pr_{x \sim \cald_0}[\indfml{\pluck(\cals)}(x) > \indfml{\cals}(x)]
            \leq
            \Pr_{x \flws \cald_0}[\forall S \in \cals, \   \indset{S \sm K}(x) = 0]
            \leq
            \Pr_{x \flws \cald_0}[\forall S \in \cals', \  \indset{S \sm K}(x) = 0]
            < \eps ,
        \end{align*}
        because $\cals'$ is a $(\cald_0,\eps)$-sunflower.
        As there are at most $n^{2w}$  iterations and $\eps = n^{-3w}$, the total error
        is at most $n^{-w}$.
    \end{proof}

    As remarked above, every input gate is already a $(w,r)$-DNF.
    Going over the gates of the circuit one-by-one (according to the circuit-order),
    we will inductively construct a $(w,r)$-DNF for each gate
    of the circuit.
    Let $g$ be a gate of the form 
    $g \defeq g_1 \circ g_2$ for a binary operation 
    $\circ \in \set{\lor,\land}$.
    Suppose we have inductively constructed
    two $(w,r)$-DNFs $\indfml{\cals}$ and $\indfml{\calt}$
    for $g_1$ and $g_2$ respectively.
    We will construct a DNF $\indfml{g}$
    such that
    \begin{align}
        E_{1,g} \defeq \Pr_{x \sim\cald_1}[\indfml{g}(x) < (\indfml{\cals} \circ \indfml{\calt})(x)] &~\leq~ (2r/q)^{w}
        \label{eq:1-err} 
        \end{align}
        and
        \begin{align}
        E_{0,g} \defeq \Pr_{x \sim\cald_0}[\indfml{g}(x) > (\indfml{\cals} \circ \indfml{\calt})(x)] &~\leq~ (2r/q)^w. 
        \label{eq:0-err} 
    \end{align}
 
    If $\circ = \lor$,
    we approximate 
    $\indfml{\cals} \lor \indfml{\calt}$
    by
    letting
    $F_g := \indfml{\pluck(\cals \cup \calt)}$.
    Note that,
    by \Cref{claim:pluck}, we have that $\indfml{\pluck(\cals \cup \calt)}$ is a $(w, r)$-DNF and
    \[\indfml{\pluck(\cals \cup \calt)} \geq 
    \indfml{\cals \cup \calt} = 
    \indfml{\cals} \lor \indfml{\calt}.\]
    Thus
    plucking introduces no errors on $\cald_1$, implying~\eqref{eq:1-err}.
    Moreover,
    plucking incurs errors at most $n^{-w} \leq (2r/q)^w$ errors on $\cald_0$
    by \Cref{claim:pluck} and $q\leq rn$, which implies \eqref{eq:0-err}. 

    If 
    $\circ = \land$,
    we approximate $g$ by first taking
    \begin{equation*}
        \calf 
        :=
        \pluck
        \left(
            \{ S \cup T : S \in \cals,\, T \in \calt \}
        \right).
    \end{equation*}
    By \Cref{claim:pluck},
    $\indfml{\calf}$
    is $r$-small since 
    $\card{S \cup T} \leq 2w$ for every $S \in \cals, T \in \calt$.
As in the $\lor$ case,
this creates no error on $\cald_1$
and the error on $\cald_0$ is at most $n^{-w} \leq (2r/q)^w$.
    We now define $\indfml{g}$ be removing all sets of width larger than $w$ from $\calf$.
    This removal does not create an error on $\cald_0$,
    since $\indfml{g} \leq\indfml{\calf}$.
    Since $\calf$ is $r$-small,
    the $q$-spreadness of $\cald_1$
    implies that
    \begin{equation*}
    \Pr_{x \sim \cald_1}[\indfml{g}(x) > \indfml{\calf}(x)] \leq
        \sum_{\ell = w+1}^{2w} 
        \sum_{S\in\calf\cap\binom{[n]}{\ell}}\Pr_{x\sim\cald_1}[\indset{S}(x)=1] 
        \leq
        \sum_{\ell = w+1}^{t} 
        r^\ell q^{-\ell}
        \leq
        (2r/q)^w,
    \end{equation*}
    where we used that $q \geq 2r$.
    We have thus shown \eqref{eq:1-err} and \eqref{eq:0-err}.
    
Finally, let $F$ be the $(w,r)$-DNF 
    of the output gate.
    By the union bound and \Cref{eq:1-err} ,
\begin{align*}
\Pr_{x \flws \cald_1}[C(x) > F(x)]
\leq \sum_{g}
E_{1,g} \leq 
   (c q/r)^w 
         (2r/q)^w 
        \leq
        (2c)^w  \leq
        0.1,
    \end{align*}
for $c = 0.05$.
A similar bound holds for $\Pr_{x \flws \cald_0}[C(x) < F(x)]$.
It follows that
\begin{equation*}
        \Pr_{x \flws \cald}[F(x) \neq f(x)]
        =
        \frac{1}{2} 
        \Pr_{x \flws \cald_1}[C(x) > F(x)] 
        + 
        \frac{1}{2} 
        \Pr_{x \flws \cald_0}[C(x) < F(x)] 
        \leq 0.1 .
        \qedhere
    \end{equation*}
\end{proof}

%

\section{Boolean separation}
\label{sec:thm2}

Recall that, given a matrix $M$ with $n$ rows and $m$ columns
over a field $\bbf$,
the Boolean function $f_M: \blt^m \to \blt$
accepts an input $S \sseq [m]$ if and only if $M[S]$ has full rank.
In this section, we only consider $\bbf \defeq \bbf_2$
in order to
prove 
\Cref{thm:main4}.

\subsection{Nonmonotone low-depth circuit upper bounds}

Let us first prove the upper bound part of Theorem~\ref{thm:main2}.
\begin{lemma}[Upper bound]
    \label{lemLowDepthUpperBoundForMatrixFunctions}
    For any matrix $M\in\Fbb^{n\times m}$, 
    the function $f_M$ is computed by a Boolean circuit of size $\poly(m)$ and depth $O(\log m)^2$.
\end{lemma}
\begin{proof}
    On an input $1_S$ corresponding to the indicator vector of a set $S\subseteq [m]$,
    we need to compute the rank of the matrix $M[S]$.
    This is known to be computed by a Boolean circuit of size $\poly(m)$ and depth $O(\log m)^2$~\cite{ABO99, Mulmuley87}. 
\end{proof}

\subsection{Monotone Boolean circuit lower bounds via well-behaved codes}
\label{sectLowerBoundF2}
\tnote{Can we give a better name for these codes with 'good' distance and dual distance?}

Let us start by defining a special kind of linear binary codes, which are explicitly constructed in Section~\ref{sectExplicitConstructionWellBehavedCodes}.
We will only use the abstract properties of these codes
to define probability distributions that allow us to apply \Cref{thm:sunflower-gen} and obtain lower bounds for monotone circuits computing functions related to these codes.

\begin{definition}
    For $M\in\Fbb^{n\times m}$,
    we define the \emph{(linear) code $\Ccal_M$ of $M$} as
    \[\Ccal_M\defeq \setbar{\transpose{M}w}{w\in\Fbb^n}\subseteq\Fbb^m,\]
    and define the \emph{dual code $\Dcal_M$ of $M$} as
    \[\Dcal_M\defeq \setbar{w\in\Fbb^m}{Mw = 0}\subseteq\Fbb^m.\]
    We say that a code $\Ccal\subseteq\Fbb^m$ is a \emph{linear code}
    if there is a rank $n$ matrix $G\in\matrixset[\Fbb]{n}{m}$, called a \emph{generator matrix} of $\Ccal$, such that 
    \[\Ccal = \Ccal_G.\]
    For any linear code $\Ccal\subseteq\Fbb^m$,
    we say that a rank $m-n$ matrix $H\in\matrixset{(m-n)}{n}$ is a \emph{parity check matrix} of $\Ccal$ if
    \[\Ccal = \Dcal_H,\]
    and, in this case,
    we define the \emph{dual code $\dual{\Ccal}$ of $\Ccal$} as
    \[\dual{\Ccal}\defeq \Ccal_H,\]
    which is independent of the choice of a parity check matrix $H$.
    We define the \emph{distance} of a linear code $\Ccal\subseteq\Fbb^m$ as 
    \[d(\Ccal)\defeq \min_{x\in\Ccal, x\neq 0}\card{\setbar{i\in [m]}{x_i\neq 0}},\]
    and its \emph{dual distance} as $\dual{d}(\Ccal)\defeq d(\dual{\Ccal})$.
    For every $d,t\in\Reals_{\geq 0}$,
    we say that $\Ccal$ is \emph{$(d,t)$-well-behaved} if
    $d(\Ccal) > d$ and $\dual{d}(\Ccal) > t$.
\end{definition}

The main theorem of this section is the following.
\begin{theorem}
    \label{theoMonotoneLowerBoundsFromGoodCodes}
    Let $n\in\Naturals$,
    and $m\defeq m(n)\in\Naturals$,
    and $d\defeq d(n)\in\Naturals$,
    and $t\defeq t(n)\in\Naturals$.
    Let $\Ccal\subseteq\Fbb^m$ be any $(d,t)$-well-behaved linear code with generator matrix $M\in\matrixset{n}{m}$.
    If $2n < \codeEscProb$,
    then the monotone complexity of $f_M$ is at least
    \[\Omega\left(\frac{\codeEscProb}{n\sqrt{t}} \right)^{\sqrt{t}/(b \log m)},\]
    \bnote{altered statement. I don't think we need $w$ here. what do you think?}
    where $b\in\Reals$ is a sufficiently large positive constant.
\end{theorem}
In Section~\ref{sectExplicitConstructionWellBehavedCodes}, we are going to use this theorem with
\[m \approx n^{3/2}\log n,\text{ and }
  d \approx m-n,\text{ and }
  t \approx n. 
\]
The rest of this section is devoted to prove this theorem.
Henceforth we assume that $\calc$ and $M$ are as in the assumption of the
theorem.
Let us denote by $M_i\in\Fbb^n$ the $i$-th column of $M$ for every $i\in[m]$.
We first prove a few properties related to
the rank of submatrices of $M$.
Some of them are standard facts from coding theory, but we add the proof here for completeness.

\begin{lemma}[Theorem 10 from Chapter 1 of \cite{MacWilliamsSloane77}]
    \label{lemSmallSetOfColsAreLinInd}
    For every $S\subseteq [m]$ with $|S|\leq \codeInd$, 
    the set of columns of $M$ indexed by $S$ is linearly independent.
\end{lemma}
\begin{proof} 
    Let $S$ be a subset of $[m]$ such that the set of columns of $M$ indexed by $S$ is linearly dependent.
    Then there is a nonzero vector $v\in\Fbb^m$ such that $\supp(v) \subseteq S$, and 
    \[\sum_{i\in [m]}v_i M_i = 0.\]
    Thus, $Mv = 0$, and $v\in\Dcal_M$.
    Note that $\Dcal_M = \dual{\Ccal}$,
    since, for any matrix $H\in\matrixset{(m-n)}{m}$ such that $\Ccal = \Dcal_H$,
    we have that 
    \[\dual{\Ccal} = \Ccal_H\text{ and }H\transpose{M} = 0.\]
    By the definition of the distance of $\dual{\Ccal}$ and the fact that $\Ccal$ is $(d,t)$-well-behaved,
    we get that
    \[|S| \geq |\supp(v)| \geq \dual{d}(\Ccal) > t.
        \qedhere
    \]
\end{proof}

\begin{lemma}
    \label{lemRandomVectorEscapesSmallSubspaces}
    For any subspace $V\subseteq \Fbb^n$ of dimension $D < n$,
    and for $u$ being a uniform random column of $M$,
    we have
    \[
    \prob{u}{u\in V} \leq 1-\codeEscProb/m \eqdef \Delta.
    \]
\end{lemma}
\begin{proof} 
    As $V$ has dimension $D\leq n-1$,
    there is a nonzero vector $c\in\Fbb^n$ such that, for every $v\in V$, we have $\transpose{v}c = 0$.
    Let $C\subseteq [m]$ be the set of the columns of $M$ that are in $V$, and let 
    \[w\defeq \transpose{M}c\in\Ccal_M = \Ccal.\]
    Note that, for every $i\in C$, we have $w_i = \transpose{M_i}c = 0$.
    Thus, by the definition of the distance of $\Ccal$ and the fact that $\Ccal$ is $(d,t)$-well-behaved, we get that
    \[m - |C| \geq |\supp(w)| \geq d(\Ccal) > d.\]
    and, consequently,
    \[
        \prob{u}{u\in V}
        \leq \frac{|C|}{m}
        \leq 1-\frac{d}{m}.
        \qedhere
    \]
\end{proof}

\begin{lemma}
    \label{lemMostLargeSetOfColsAreFullRank}
    For any $N\defeq N(m)\in\Naturals$ with $N\leq m$, and
    for a random set $S$ distributed uniformly in $\binom{[m]}{N}$,
    we have
    \[\prob{S}{\text{$M[S]$ is not full rank}} \leq 2^n\Delta^N.\]
\end{lemma}
\begin{proof}
    Let $\linspan(S)\defeq \linspan(\setbar{M_i}{i\in S})$, where $M_i$ is the column of $M$ indexed by $i\in [m]$.
    Note that the rank of $M[S]$ is at most $n-1$ if and only if
    there is a subspace $V\subseteq \Fbb^n$ of dimension $n-1$ such that $\linspan(S)\subseteq V$.
    Hence,
    \begin{align*}
        \prob{S}{\text{$M[S]$ is not full rank}}
        &\leq \prob{S}{\exists V\subseteq \F_2^n\sttext\dim(V)=n-1\text{ and }\linspan(S)\subseteq V} \\
        &\leq \sum_{V\subseteq \F^n_2, \dim(V)=n-1} \prob{S}{S\subseteq V}.
    \end{align*}
    By Proposition 1.7.2 in Stanley~\cite{StanleyVol1SecondEd97}, the number of $n-1$ dimensional subspaces of $\Fbb^n$ is 
    \[
        \frac{(2^n-1)(2^n-2)\cdots (2^n-2^{n-2})}
        {(2^{n-1}-1)(2^{n-1}-2)\cdots (2^{n-1}-2^{n-2})}
        \leq 2^n-1.  
    \]
    For any $(n-1)$-dimensional subspace $V$ of $\Fbb^n$, let $C_V\subseteq [m]$ be the set of columns of $M$ that are contained in $V$.
    By Lemma~\ref{lemRandomVectorEscapesSmallSubspaces}, we know that $c\defeq |C_V|\leq \Delta m$, which implies that
    \begin{align*}
        \prob{S}{S\subseteq V}
        \leq \prob{S}{S\subseteq C_V}
        = \frac{\binom{c}{N}}{\binom{m}{N}}
        \leq \left(\frac{c}{m}\right)^N
        \leq \Delta^N.
    \end{align*}
    Therefore, 
    \begin{equation*}
        \prob{S}{\text{$M[S]$ is not full rank}}
        \leq 2^n\Delta^N.
        \qedhere
    \end{equation*}
\end{proof}

Now our goal is to apply \Cref{thm:sunflower-gen} to prove a lower bound for $f_M$.
So we first define two probability distributions that, intuitively, are hard to distinguish for monotone Boolean circuits of small size.
\begin{definition}
    \label{defDistributionAndError}
    We define the \emph{distribution $a\sim D_1$} by sampling a uniformly
    random $a\in \{0,1\}^m$ of Hamming weight 
    \begin{equation}
        \label{eqChoiceOfHammingWeightN}
        N\defeq n\ceiling{m/\codeEscProb}.
    \end{equation}
    We define the \emph{distribution $a\sim D_0$} by sampling a uniformly
    random $u\in \F_2^n$ and, for every $j\in [m]$, we set
    $a_j \defeq 1$ if $\iprod{M[j]}{u} = 0$
    and
    $a_j \defeq 0$ otherwise, where the inner product is taken over $\F_2$.
\end{definition}

    Note that, by the choice of $N$ and the assumption $2n\leq \codeEscProb$, we have $N < m$ and
    \[1-2^n(1-\codeEscProb/m)^N
      \geq 1-2^ne^{-N\codeEscProb/m}
      \geq 1-2^n2^{-n/\ln 2}
      \geq 1/2.\]
    Thus, by Lemma~\ref{lemMostLargeSetOfColsAreFullRank},
    we have
    \[
        \prob{a\sim D_1}{f_M(a) = 1} \geq 1/2 \eqdef \alpha.
    \]
    Moreover,
    the point $a$ sampled from $D_0$ is an element of
    $f_M^{-1}(0)$ as long as $u \neq 0$,
    so
    \[
    \prob{a\sim D_0}{f_M(a) = 0} = 1 - \frac{1}{2^n}\geq \alpha.
    \]

\paragraph{Spreadness of $D_1$.}

To apply \Cref{thm:sunflower-gen}, we need to show that $D_1$
is ``spread''.

\begin{lemma}[Spreadness of $D_1$]
\label{lem:D1sparse-F2}
The distribution $D_1$ is $N$-wise $(m/N)$-spread.
\end{lemma}

\begin{proof}
Let $T$ be any subset of $[m]$ of size $k\leq N$. 
The number of $N$-sized sets of columns of $M$ which contain the set of columns
indexed by $T$ is $\binom{m-k}{N-k}$. Thus, 
\[
    \prob{a\sim D_1}
    {\bigwedge_{i\in T}a_i = 1}
    \leq \frac{\binom{m-k}{N-k}}{\binom{m}{N}}.
\]
Simplifying the above, we get
\[
    \prob{a\sim D_1}
    {\bigwedge_{i\in T}a_i = 1}
    \leq \left(\frac{\prod_{i=0}^{k-1} (N-i)}{\prod_{i=0}^{k-1}(m-i)}\right)
    \leq \left(\frac{N}{m}\right)^k, 
\]
using the fact that $\frac{N-i}{m-i}\leq \frac{N}{m}$ for $i\leq k-1$. 
\end{proof}

\paragraph{Sunflower bound for $D_0$.}
The second step is to show a sunflower bound for $(D_0,\eps)$-sunflowers.
We first prove that $D_0$ is $\codeInd$-wise independent. 

\begin{lemma}[Independence of $D_0$]
    \label{lem:D0ind-F2}
    The distribution $D_0$ has uniform marginals and is $\codeInd$-wise independent.
\end{lemma}

\begin{proof}
    Fix any set $S\defeq\{j_1, j_2, \ldots, j_\codeInd\}$ of $[m]$. Also fix a vector 
    \[
        b\defeq (b_1, b_2, \ldots, b_\codeInd)\in \F_2^\codeInd.
    \]
    Then
    \[
        \prob{a\sim D_0}{a\restrict{S} = b}
        = \prob{u\in \F_2^n}{(\iprod{u}{M[j_1]} =b_1) \wedge \ldots \wedge (\iprod{u}{M[j_\codeInd]}=b_\codeInd)}
        = \prob{u\in\F_2^n}{\transpose{u}M[S]=b}.
    \]
    Using the rank-nullity theorem, 
    the above probability is the same as 
    \[
        \prob{u\in\F_2^n}{\transpose{u}M[S]=b}
        = \frac{|\ker(M[S])|}{2^n}
        = \frac{1}{2^\codeInd},
    \]
    where the last equality follows from the fact that the rank of $M[S]$ is $\codeInd$ (Lemma~\ref{lemSmallSetOfColsAreLinInd}).
    As the columns of $M$ are nonzero (as $\codeInd > 0$), we can prove that the marginals are uniform: for every $\ell\in [\codeInd]$,
    \[
        \prob{u\in\F_2^n}{\iprod{u}{M[j_{\ell}]}=0}=\frac{1}{2}.
        \qedhere
    \]
\end{proof}

Another important ingredient in our proof is Bazzi's
theorem~\cite{Bazzi09}, which has been further improved by work of
Tal~\cite{Tal17}.
This shows that small DNFs have almost the same behaviour on $k$-wise independent
distributions as they have over truly uniform distributions.

\begin{theorem}[Fooling DNFs with independence: Theorem 7.1~\cite{Tal17}]
    \label{theoBazziTheorem}
    Let $F$ be a DNF with $N$ terms and $\varepsilon\in (0,1)$.
    Then there exists $T \leq O(\log N\cdot \log(N/\varepsilon))$ such that,
    for any distribution $\Dcal$ that is $T$-wise independent with uniform marginals, we have
    \[
    \left|\prob{a\sim \{0,1\}^n}{F(a) = 0} - \prob{a\sim \Dcal}{F(a) = 0} \right|\leq \varepsilon.
    \]
\end{theorem}

Finally, to show our sunflower lemma,
we combine Bazzi's theorem with the optimal bounds for ``robust
sunflowers''
recently obtained 
(see, e.g.,~\cite[Lemma 2]{Rao2025}, \cite{BellChueluechaWarnke2021}).
\emph{Robust sunflowers} are $(\mathsf{Uniform}, \eps)$-sunflowers,
where $\mathsf{Uniform}$ denotes the uniform distribution.
Recall the definition of ``sunflower bound'' $r(\cdot, \cdot, \cdot)$ from
\Cref{def:abstract-sunflower}.

\begin{lemma}[Robust Sunflower Lemma: Lemma~2~\cite{Rao2025}]
    \label{lem:RSL}
    For every $k\in\Naturals$ and $\eps\leq 1/2$, 
    we have 
    $r(\mathsf{Uniform}, k, \eps) 
    = O(\log(k/\eps))$.
\end{lemma}

We can now show a sunflower lemma for $D_0$.
Recall that $m$ is the number of input bits of the function $f_M$.
\begin{lemma}[Sunflower lemma for $t$-wise independent distributions]
    \label{lem:F2-neg-sunflower}
    There exists a constant $b \geq 1$,
    such that, for
    $w \defeq \sqrt{\codeInd}/(b \log m)$
    and every 
    $k \leq 2w$,
    we have
    \[r(D_0,k,m^{-4w}) = O(w \log m).\]
\end{lemma}
\begin{proof}
    Let $w \defeq \sqrt{\codeInd}/(b \log m)$ for some constant $b$ to be defined
    later, and set $\eps \defeq m^{-4w}$.
    We will show 
    \[r(D_0,k,\eps) \leq r(\mathsf{Uniform}, k, \eps/2).\]
    Indeed, suppose that $\cals$ is an $\eps/2$-robust sunflower with core
    $K$.
    Let $\cals_K = \setbar{S \sm K}{S \in \cals}$.
    Note that $\indfml{\cals_K}$ has width at most $k$ and at most $\binom{m}{\leq k}\leq m^k$ terms.

    We now wish to apply
    \Cref{theoBazziTheorem} on $\indfml{\cals_K}$ with the 
    $\codeInd$-wise independent distribution $D_0$ and 
    approximation parameter $\eps/2$.
    \Cref{theoBazziTheorem} can only be applied when
    \begin{equation*}
        \codeInd \geq D \log (\card{\cals_K}) \log(\card{\cals_K}/(\eps/2)),
    \end{equation*}
    for some constant $D > 0$ given by \Cref{theoBazziTheorem}.
    Now note that
    \begin{equation*}
        D \log (\card{\cals_K}) \log(\card{\cals_K}/(\eps/2))
        \leq O(w^2 (\log m)^2)
        = \frac{1}{b^2} O(\codeInd)
        \leq \codeInd
    \end{equation*}
    for large enough $b$.
    Thus, we can apply \Cref{theoBazziTheorem} and obtain
    \begin{align*}
        \Pr_{x \flws D_0}[\exists S \in \cals : \indset{S \sm K}(x) = 1]
        &= \Pr_{x \flws D_0}[\indfml{\cals_K}(x) = 1] \\
        &\geq \Pr_{x \flws \blt^n}[\indfml{\cals_K}(x) = 1] - \eps/2 \\
        &= \Pr_{x \flws \blt^n}[\exists S \in \cals : \indset{S \sm K}(x) = 1] - \eps/2 \\
        &> 1-\eps.
    \end{align*}
    Thus, the set $\cals$ is a $(D_0,\eps)$-sunflower.
    The result now follows from \Cref{lem:RSL} by observing that
    \begin{equation*}
        r(\mathsf{Uniform},k,\eps/2)
        \leq O(\log(k/\eps))
        \leq O(w \log m).
        \qedhere
    \end{equation*}
\end{proof}

\paragraph{Wrapping up.}

We can now apply \Cref{thm:sunflower-gen}, finishing the proof.
Note that the distributions $D_1$ and $D_0$
defined in \Cref{defDistributionAndError}
satisfy the following properties:
\begin{enumerate}
    \item We have
        $\Pr_{x \flws D_i}[f_M(x)=i] 
        \geq \alpha$
        for every $i \in \blt$;
    \item $D_1$ is $N$-wise $(m/N)$-spread (\Cref{lem:D1sparse-F2});
    \item For $w \defeq \sqrt{\codeInd}/(b \log m)$
        where $b$ is a sufficiently large universal constant (Lemma~\ref{lem:F2-neg-sunflower}),
        we have $r(D_0, k, \alpha m^{-3w})\leq r(D_0, k, m^{-4w})\leq O(w \log m)$
        for every $k \leq 2w$ and sufficiently large $n$.
\end{enumerate}
Therefore, by \Cref{thm:sunflower-gen} and the choice of $N$ (Equation~\ref{eqChoiceOfHammingWeightN}),
we obtain that the monotone complexity of $f_M$ is at least
\begin{equation*}
    \Omega\left(\frac{m}{Nw \log m} \right)^w
    \geq \Omega\left(\frac{\codeEscProb}{nw \log m} \right)^w
    =
    \Omega
    \left( 
        \frac{d}{n \sqrt{t}}
    \right)^{\sqrt{t}/(b \log m)}.
    \qedhere
\end{equation*}
\bnote{Need to choose $w$.}\tnote{$w$ is chosen in the bullet point above}
This finishes the proof of Theorem~\ref{theoMonotoneLowerBoundsFromGoodCodes}.

\subsection{Explicit constructions of well-behaved codes} 
\label{sectExplicitConstructionWellBehavedCodes}

We now construct an explicit well-behaved binary code $\Ccal$ for which we can apply the argument from Section~\ref{sectLowerBoundF2} to prove exponential lower bounds for monotone circuits computing a boolean function in uniform-$\NC^2$.

\begin{lemma}
    \label{lemmaTransformationOfGoodCodesFromLargeAlphabetToBinary}
    Let $n,m,l\in\Naturals$, and let $\Fbb_q$ be a finite field with $q\defeq 2^l$.
    Let $\Lcal\subseteq \Fbb_q^m$ be the linear code generated by a
    matrix $G\in\matrixset[\Fbb_q]{n}{m}$.
    Then there is a linear binary code $\Ccal\subseteq\Fbb_2^s$ generated by a matrix
    $M\in\matrixset[\Fbb_2]{r}{s}$ for $r\defeq ln$ and $s\defeq lm$ such that
    \[       d(\Lcal) \leq        d(\Ccal) \leq l       d(\Lcal)\text{ and } 
      \dual{d}(\Lcal) \leq \dual{d}(\Ccal) \leq l\dual{d}(\Lcal). 
    \]
    Furthermore, if each entry of the generator matrix $G$ for $\Lcal$ can be computed in time $\poly(n, m, l)$,
    then we can compute all the entries of the generator matrix $M$ for $\Ccal$ in time $\poly(n, m, l)$.
\end{lemma}
\begin{proof}
    Let $\Bcal\defeq\set{b_1,\dotsc, b_l}$ be a basis of $\Fbb_q$ as an
    $l$-dimensional vector space over $\Fbb_2$.
    Let $\phi\ffrom \Fbb_q\fto\Fbb_2^l$ be the map
    from elements $v\in\Fbb_q$ to their coordinate vector with respect to the basis $\Bcal$:
    that is, for every $v\in\Fbb_q$, let $\phi(v)\in\Fbb_2^l$ be the unique vector such that
    \[v = \sum_{i\in[l]}\phi(v)_ib_i.\]
    Note that $\phi$ is a linear isomorphism between $\Fbb_q$ and $\Fbb_2^l$.
    Let $\gamma\ffrom \Fbb_q^m\fto\Fbb_2^{lm}$ be the following map: for every $v_1,\dotsc,v_m\in\Fbb_q$,
    \[\gamma(v_1,\dotsc,v_m)\defeq (\phi(v_1),\dotsc, \phi(v_m)).\]
    Again note that $\gamma$ is a linear isomorphism between $\Fbb_q^m$ and $\Fbb_2^{lm}$; in particular:
    \[\gamma(0) = 0\text{ and $\gamma$ is a linear bijection}.\]
    Let
    \[\Ccal\defeq\setbar{\gamma(w)}{w\in\Lcal},\]
    and let $g_1,\dotsc,g_n\in\Fbb^m$ be the rows of $G$.
    Note that, for every $w\in\Fbb_q^k$,
    \[\transpose{G}w
    = \sum_{i\in[n]}w_ig_i
    = \sum_{i\in[n]}\sum_{j\in[l]}\phi(w_i)_jb_jg_i,
    \]
    hence 
    \[\gamma(\transpose{G}w)
    = \sum_{i\in[n]}\sum_{j\in[l]}\phi(w_i)_j\gamma(b_jg_i).
    \]
    Let $M\in\matrixset[\Fbb_2]{r}{s}$, for $r\defeq ln$ and $s\defeq lm$,
    be the matrix with the vectors $\gamma(b_jg_i)$ as rows, that is,
    \begin{equation}\label{eqGeneratorMatrixForBinaryGoodCode}
        M \defeq
        \sum_{i\in[n]}\sum_{j\in[l]}e_{(i-1)l + j}\transpose{\gamma(b_jg_i)}
        =
        \begin{bmatrix}
            \transpose{\gamma(b_1g_1)} \\
            \vdots \\
            \transpose{\gamma(b_ng_l)}
        \end{bmatrix},
    \end{equation}
    where $e_1,\dotsc,e_{nl}$ are the standard basis vectors of $\Fbb_2^{nl}$.
    As $\phi$ is a linear isomorphism, we can prove that
    \[\Ccal = \setbar{\transpose{M}w}{w\in\F_2^{nl}}.\]
    Now let us prove the claimed distance bounds for $\Ccal$.
    First note that if $v\in\Ccal$ satisfies $v = \gamma(w) = (\phi(w_1), \dotsc, \phi(w_m))$ for some $w\in\Lcal\subseteq\Fbb_q^m$,
    then, as $\phi(x) = 0$ iff $x = 0$, we have
    \[|\supp(w)|\leq |\supp(v)|,\]
    which implies that $d(\Lcal)\leq d(\Ccal)$,
    and $|\supp(v)|\leq l|\supp(w)|$, which implies that
    \[d(\Ccal)\leq ld(\Lcal).\]
    In order to prove the bound on the dual distance of $\Ccal$, we will use the following characterization of the dual distance of codes:
    \begin{proposition}[Theorem 8 from Chapter 5 of~\cite{MacWilliamsSloane77}] 
        \label{propDualDistanceAndKWiseIndependence}
        We say that a set $S\subseteq\Fbb_q^m$ is $t$-wise independent if,
        for a uniformly chosen $(x_1,\dotsc,x_m)\defeq x\sim\Unif(S)$ and
        for all $I\subseteq [m]$ with $|I|\leq t$,
        the variables $(x_i)_{i\in I}$ are independent and uniformly distributed over $\Fbb_q$.
        Then a linear code $\Ccal\subseteq\Fbb_q^m$ is $t$-wise independent if and only if $t\leq \dual{d}(\Ccal) - 1$.
        \bnote{If we are invoking this theorem, perhaps we could avoid
        proving \Cref{lem:D0ind-F2}, which is just a special case of this one.}
    \end{proposition}
    Now let us prove that $\dual{d}(\Lcal)\leq\dual{d}(\Ccal)$.
    For any $t\leq \dual{d}(\Lcal)-1$, we know that $\Lcal$ is $t$-wise independent over $\Fbb_q$ by Proposition~\ref{propDualDistanceAndKWiseIndependence}.
    This implies that,
    for a uniformly chosen $(x_1,\dotsc,x_m)\defeq x\sim\Unif(\Lcal)$ and
    for $I\subseteq [m]$ with $|I|\leq t$,
    the variables $(x_i)_{i\in I}$ are independent and uniformly distributed over $\Fbb_q$.
    Let
    \[(y_1,\dotsc,y_{lm})\defeq y\sim\Unif(\Ccal).\]
    As $x_i$ for any $i\in I$ is uniformly over $\Fbb_q$ and as $\phi$ is a bijection from $\Fbb_q$ to $\Fbb_2^l$,
    we get that $(y_{i,1},\dotsc, y_{i,l})\defeq\phi(x_i)$ are independent random variables and uniformly distributed over $\Fbb_2$.
    Thus, any subset $J\subseteq I\times [l]$ of the random variables $(y_{i,p})_{i\in I, p\in [l]}$ with $|J|\leq t$ satisfies that
    $(y_j)_{j\in J}$ are independent random variables and uniformly distributed over $\Fbb_2$.
    Hence, $\Ccal$ is $t$-wise independent, and, by Proposition~\ref{propDualDistanceAndKWiseIndependence},
    \[\dual{d}(\Lcal)\leq \dual{d}(\Ccal).\]
    Now suppose that 
    \[\dual{d}(\Ccal) \geq l\dual{d}(\Lcal) + 1.\]
    By Proposition~\ref{propDualDistanceAndKWiseIndependence}, we get that $\Ccal$ is $(l\dual{d}(\Lcal))$-wise independent.
    Again using the fact $\phi$ is a bijection from $\Fbb_q$ to $\Fbb_2^l$,
    we can prove that,
    for a uniformly chosen $(x_1,\dotsc,x_m)\defeq x\sim\Unif(\Lcal)$ and
    for $I\subseteq [m]$ with $|I|\leq \dual{d}(\Lcal)$,
    the variables $(x_i)_{i\in I}$ are independent and uniformly distributed over $\Fbb_q$.
    By Proposition~\ref{propDualDistanceAndKWiseIndependence},
    we obtain that $\dual{d}(\Lcal) + 1\leq \dual{d}(\Lcal)$,
    which is a contradiction.
    Therefore,
    \[\dual{d}(\Ccal) \leq l\dual{d}(\Lcal).\]

    Now let us argue about the explicitness of our construction.
    Note that:
    \begin{itemize}
        \item Element representations (via a basis for $\Fbb_q$) and operations over the field $\Fbb_q$
            can be performed in time $\poly(l)$~\cite{Shoup88}.
        \item For any given $(v_1,\dotsc,v_m)\defeq v\in\Fbb_q^m$,
            the element $\gamma(v)\in\Fbb_2^{lm}$ can be uniformly computed in time $\poly(m, l)$
            as it is just the concatenation of $\phi(v_1), \dotsc, \phi(v_m)\in\Fbb_2^l$ and
            $\phi(v_i)$ is the representation of $v_i$ as an element of $\Fbb_q$.
        \item The generator matrix $M\in\matrixset[\Fbb_2]{ln}{lm}$ (Equation~\ref{eqGeneratorMatrixForBinaryGoodCode})
            for $\Ccal\subseteq\Fbb_2^{lm}$ can be uniformly computed in time $\poly(n, m, l)$.
            \qedhere
    \end{itemize}
\end{proof}

\begin{corollary}
    \label{coroWellBehavedBinaryCodesFromReedSolomon}
    For every $n,m\in\Naturals$ with $m > n$ and for $l\defeq\ceiling{\log_2 m}$,
    there is a binary code $\Ccal\subseteq\Fbb_2^{lm}$ such that
    we can explicitly construct
    a generator matrix $M\in\matrixset[\Fbb_2]{ln}{lm}$ for $\Ccal$
    and
    \[m - n\leq d(\Ccal)        \leq 2l(m-n)\text{ and }
      n    \leq \dual{d}(\Ccal) \leq 2ln.\]
\end{corollary}
\begin{proof}
    For $q\defeq 2^l$, let $\Lcal\subseteq\Fbb_q^m$ be a Reed-Solomon code over $\Fbb_q$ with generator matrix $G\in\matrixset[\Fbb_q]{n}{m}$. 
    The following facts are standard results in coding theory (e.g., see Chapter 10 of~\cite{MacWilliamsSloane77}):
    \begin{itemize}
        \item Each entry of the matrix $G$ can be constructed in time $\poly(n,m,l)$.
        \item The distance of $\Lcal$ is 
            \[d(\Lcal) = m-n+1.\]
        \item The dual code of $\Lcal$ 
            has distance
            \[\dual{d}(\Lcal) = m - (m - n) + 1 = n + 1.\]
    \end{itemize}
    By Lemma~\ref{lemmaTransformationOfGoodCodesFromLargeAlphabetToBinary},
    there is a binary code $\Ccal\subseteq\Fbb_2^{lm}$ such that
    we can explicitly construct
    a generator matrix $M\in\matrixset[\Fbb_2]{ln}{lm}$ for $\Ccal$,
    and
    \[m-n\leq d(\Lcal) \leq d(\Ccal) \leq ld(\Lcal)\leq 2l(m-n),\text{ and }
      n\leq \dual{d}(\Lcal) \leq \dual{d}(\Ccal) \leq l\dual{d}(\Lcal)\leq 2ln.
      \qedhere
  \]
\end{proof}

\begin{theorem}[Theorem~\ref{thm:main4}]
    For every sufficiently large $n\in\Naturals$, and
    for
    \[m\defeq \ceiling{n^{3/2}(\log n)^2}\text{ and }
      l\defeq \ceiling{\log_2 m},\]
    there is an $lm$-variate monotone Boolean function $f\ffrom\bcube{lm}\fto\bset$ such that the following hold:
    \begin{enumerate}
        \item $f$ can be computed by a uniform Boolean circuit of size polynomial in $n$
            and depth $O(\log lm)^2$.
        \item Any monotone Boolean circuit computing $f$ has size at least $2^{(lm)^{1/3-o(1)}}$.
    \end{enumerate}
\end{theorem}
\begin{proof}
    By Corollary~\ref{coroWellBehavedBinaryCodesFromReedSolomon},
    there is a binary code $\Ccal\subseteq\Fbb_2^{lm}$ such that
    we can explicitly construct
    a generator matrix $M\in\matrixset[\Fbb_2]{ln}{lm}$ for $\Ccal$
    and
    \begin{equation}
        \label{eqBoundsOnDistancesOfGoodCode}
        m-n\leq d(\Ccal) \leq 2l(m-n)\text{ and }
        n \leq \dual{d}(\Ccal)\leq 2ln.
    \end{equation}
    Let
    $f\defeq f_M$ be the $lm$-variate Boolean function corresponding to $M$.
    By Lemma~\ref{lemLowDepthUpperBoundForMatrixFunctions}, the function $f$
    can be computed by uniform Boolean circuit of size $\poly(lm)$ and depth $O(\log lm)^2$.
    Let 
    \[d\defeq d(\Ccal)-1\text{ and } t\defeq \dual{d}(\Ccal)-1.\] 
    Note that, for sufficiently large $n$,
    \[d\geq 2nl.\]
    By Theorem~\ref{theoMonotoneLowerBoundsFromGoodCodes}, we obtain that
    the monotone complexity $S^+(f)$ of $f$ is at least
    \[\Omega\left(\frac{\codeEscProb}{lnw \log (lm)} \right)^w,\]
    for $w \defeq \sqrt{\codeInd}/(b \log (lm))$.
    By Equation~\ref{eqBoundsOnDistancesOfGoodCode}, we get that
    \[\Omega(m^{1/3}/(\log m)^2)
      \leq w
      \leq O(n^{1/2}/(\log m)^{1/2}).
    \]
    Therefore,
    \[S^+(f)
    \geq \Omega\left(\frac{\codeEscProb}{lnw \log (lm)} \right)^w
    \geq \Omega\left(\frac{m}{n^{3/2}(\log m)^{3/2}} \right)^w
    \geq \Omega(\log m)^{w/2}
    \geq 2^{\Omega(m^{1/3}\log\log m/(\log m)^2)}.
    \qedhere
    \]
\end{proof}

\section{Mixed separation}
\label{sec:thm3}

This section uses some ideas from Section~\ref{sec:thm2} to prove a lower bound for $f_M$ when $M$ is a real matrix. We get weaker (but still superpolynomial) lower bounds for the size of monotone circuits.

\subsection{Nonmonotone low-depth circuit upper bounds}

Let us first prove the upper bound part of Theorem~\ref{thm:main3}.
\begin{lemma}[Upper bound]
    \label{lemmaAlgebraicUBReals}
    For any matrix $A\in\Reals^{n\times m}$,
    the polynomial $P_A\in \mathbb{R}[x_1,\ldots, x_m]$, as defined in Equation~\ref{eqDefinitionHomLinPoly},
    is computed by an arithmetic circuit of size $\poly(m)$ and depth $O(\log m)^2.$
\end{lemma}
\begin{proof}
    Let 
    \[I_X \defeq\sum_{i=1}^m x_ie_i\transpose{e_i}\in\Reals[X]^{m\times m}\]
    be the $m$-dimensional identity matrix with diagonal elements replaced by the variables
    $x_1,\dots, x_m$.
    By the Cauchy-Binet formula, we have
    \[\det((AI_X)A^T)
      = \sum_{\substack{S\subseteq [m]:\\ |S| = n}} \det((AI_X)[[n],S])\det(A^T[S,[n]])
      = \sum_{\substack{S\subseteq [m]:\\ |S| = n}} \det((AI_X)[[n],S])\det(A[S]),\]
    and, by the block structure of $I_X$ and the multiplicative property of the determinant,
    \[\det((AI_X)[[n], S])
      = \det(A[[n], S]I_X[S, S])
      = \det(A[[n], S])\det(I_X[S, S])
      = \det(A[S])\prod_{i\in S}x_i.\]
    Hence,
    \[\det((AI_X)A^T)
      = \sum_{\substack{S\subseteq [m]:\\ |S| = n}} \left(\det(A[S])\prod_{i\in S}x_i\right)\det(A[S])
      = P_A(x_1,\dotsc, x_m).\]
    By efficient computation of the determinant of a symbolic matrix~\cite{Berkowitz84, MV97}, we obtain that the polynomial
    $P_A = \det(AI_XA^T)$ can be computed by an arithmetic circuit of size $\poly(m)$ and depth $O((\log m)^2)$.
\end{proof}

\subsection{Choice of a well-behaved $\Reals$-matrix}

Let $k\defeq k(n)\in\Naturals$ be a growing function, which will be specified later, such that $k\leq n^{0.1}$.
Let 
\begin{equation}
    \label{eq:params-reals}
    m\defeq n^2
    \quad
    \quad
    \text{and}
    \quad
    \quad
    s\defeq 200k^2.
\end{equation}
\bnote{It would be nice to have a discussion for why we need this value of
$s$.}
Let $M_n$ be an $n\times m$ random matrix obtained by sampling each column
$M_n[i]$ of $M_n$ independently and uniformly at random from the set of
vectors in $\{0,1\}^n$ of Hamming weight at most $s$.

For technical reasons, it is nicer to sample the matrix using points of Hamming weight \emph{at most} $s$, though the following lemma shows that this is not very different from sampling points of weight close to $s.$
\begin{lemma}
    \label{lem:M-suppsize}
    With probability at least $1-1/n,$ every column of $M_n$ has Hamming weight at least $s/2.$ 
\end{lemma}

\begin{proof}
    Note that $s = 200k^2 = O(n^{0.2})$.
    We note that $\binom{n}{i} \leq n^{i}$ for $i\in [n]$, and use the bound on $s$ to conclude
    \[
    \binom{n}{\leq s} \geq \binom{n}{s} \geq \left(\frac{n}{s}\right)^s \geq n^{0.7s} .
    \]
    Putting things together, the probability that any fixed column $M_n[j]$ has Hamming weight less than $s/2$ is at most
    \[
        \sum_{i=0}^{s/2} \frac{\binom{n}{i}}{\binom{n}{\leq s}}
        \leq \frac{(s/2+1)\cdot n^{s/2}}{n^{0.7s}}\leq \frac{n}{n^{0.2s}}.
    \]
    Union bounding over $m = n^2$ columns and using the fact that $s\geq 200$,
    the probability that there is a column of Hamming weight at most $s/2$ is at most $1/n$.
\end{proof}

To construct our hard distributions, the following lemma will be important. 
\begin{lemma}[Full rank lemma]
    \label{lem:M-prop1-R}
    If $S$ is distributed uniformly in ${[m] \choose 10n\log n}$ and independently of $M_n$, then 
    \[
        \prob{M_n}{\prob{S}{\text{$M_n[S]$ is full rank}} \geq 1/10} \geq 1/10.
    \]
\end{lemma}
\begin{proof}
    The idea is to carry out an argument analogous to the solution to the coupon-collector's problem to show that for each $S\in \binom{[m]}{10n\log n}$, we have
    \begin{equation}
        \label{eq:Mn-rand-S} 
        \prob{M_n}{\text{$M_n[S]$ is full rank}} \geq \frac{4}{5}.
    \end{equation}
    The claim of the lemma then follows from Equation~\ref{eq:Mn-rand-S} by the following averaging argument.
    Define the random variables
    \[X \defeq \prob{S\in \binom{[m]}{n}}{\text{$M_n[S]$ is full rank}},\text{ and }
    Y\defeq 1 - X.
    \]
    Thus
    \begin{equation*} 
        \begin{split}
            \mathbb{E}_{M_n}[X]
            &= \mathbb{E}_{M_n}\left[ \frac{1}{\binom{m}{n}} \sum_{S\subseteq [m] : |S|=n} \mathbbm{1}_{\rank(M_n[S])=n}\right] \\
            &= \frac{1}{\binom{m}{n}}\sum_{S\subseteq [m] : |S|=n} \prob{M_n}{\rank(M_n[S])=n}
            \geq \frac{4}{5},
        \end{split}
    \end{equation*}
    and $\mathbb{E}_{M_n}[Y]\leq \frac{1}{5}$.
    By Markov's inequality, we get
    \[
        1-\prob{M_n}{X\geq \frac{1}{10}}
        = \prob{M_n}{Y\geq \frac{9}{10}}
        \leq \frac{10\Exp_{M_n}[Y]}{9}
        \leq \frac{2}{9},
    \]
    which implies that
    \[\prob{M_n}{X\geq \frac{1}{10}}\geq \frac{7}{9} > \frac{1}{10}.\]

    Now let us prove the probability bound in Equation~\ref{eq:Mn-rand-S}.
    Assume, without loss of generality, that $S = \{1,\ldots, 10n\log n\}$.
    Assume that we choose an infinite sequence of vectors $\set{v_i}_{i\in\Naturals}$ of vectors independently and uniformly at random $\{0,1\}^n_{\leq s}$.
    For each $d \leq n$, let $X_d$ denote the smallest $r$ such that $v_1,\ldots, v_r$ span a subspace of dimension $d$.
    It suffices to show that $\mathbb{E}X_n \leq 2n\log n$, because Equation \ref{eq:Mn-rand-S} then follows by Markov's inequality,
    since we can interpret the definition of $M_n$ as choosing the vectors $v_1,\cdots, v_{10n\log n}$ as $M_n$'s columns.
    For $X_0 \defeq 0$, we have that $X_n \defeq \sum_{d=0}^{n-1} X_{d+1}-X_{d}$, which implies that
    \[\mathbb{E}X_n = \sum_{d=0}^{n-1} \mathbb{E}[X_{d+1}-X_{d}].\]
    To compute $\mathbb{E}[X_{d+1}-X_{d}]$, we condition on the value $X_d=r$ and the vectors $v_1,\ldots, v_r$ which span a vector space $V$ of dimension $d.$ Let $p_d$ denote the probability that a uniformly random $v\in \{0,1\}^n_{\leq s}$ lies outside $V$. By the subspace lemma (Lemma~\ref{lem:intersect-space-R}) proved below, we have
    \[
    1-p_d \leq \frac{\binom{d}{\leq s}}{\binom{n}{\leq s}} \leq \frac{d+1}{n+1}
    \]
    where the latter inequality follows from the following simple binomial estimate:
    \[
    \binom{d}{\leq s} = 1 + d + \sum_{i=2}^s\binom{d}{i} \leq 1+ d+\sum_{i=2}^s\frac{d+1}{n+1}\cdot \binom{n}{i}=\frac{d+1}{n+1}\cdot \binom{n}{\leq s}.
    \]
    It follows that $X_{d+1}-X_d$ (conditioned on $v_1,\ldots, v_r$) has a geometric distribution with success probability $p_d\geq (n-d)/(n+1)$ and thus has expectation at most $(n+1)/(n-d)$. 
    Hence
    \[
    \mathbb{E}X_n \leq \sum_{d=0}^{n-1} \frac{n+1}{n-d} \leq 2n\log n. 
    \]
    We have thus shown the desired upper bound on $\mathbb{E}X_n$, which completes the proof of the lemma.
\end{proof}

The above proof used the following lemma, which show that a uniformly randomly chosen $v\in \{0,1\}^n_s$ cannot lie in any proper subspace of $\mathbb{R}^n$ with high probability.

\begin{lemma}[Subspace lemma]\label{lem:intersect-space-R}
    Let $\F$ be any field and let $V$ be a subspace of $\F^n$ of dimension $d$, and $\{0,1\}^n_{\leq s}$ be the set of all binary strings of Hamming weight at most $s$. Then, 
    \[
        \lvert V\cap\{0,1\}^n_{\leq s} \lvert \leq \binom{d}{\leq s}.
    \]
\end{lemma}

We note that the above lemma is tight, as witnessed by a subspace $V$ generated by any $d$ standard basis vectors.
\begin{proof}
    As $V$ has dimension $d$,
    there is a basis $B$ of $V$ such that
    there is a set $R\in\binom{[n]}{d}$ indexing the elements of $B$ as
    \[B = \setbar{b^{(r)}}{r\in R} \]
    such that, for every $r\in R$ and $v\in B$,
    \[\begin{cases}
        v_r = 1 & \text{if $v = b^{(r)}$, and} \\
        v_r = 0 & \text{otherwise.}
      \end{cases}\]
    Note that we can find such a basis via Gaussian elimination over any arbitrary basis of $V$,
    as the set $R$ above corresponds to the rows of the pivot element of the matrix in column echelon form.
    For every $v\in X\defeq V\cap \{0,1\}^n_{\leq s}$, let $\alpha_v\in\Reals^B$ be the unique vector such that
    $v = \sum_{b\in B}\alpha_{v, b}b$.
    So, for every $r\in R$,
    \[v_r
    = \sum_{b\in B} \alpha_{v,b}b_r
    = \alpha_{v,b^{(r)}},\]
    which implies that $\supp(v)\cap R$ uniquely determines $\alpha_v$ (as $v_r\in\bset$),
    and, consequently, uniquely determines $v$.
    As $|\supp(v)\cap R|\leq |\supp(v)|\leq s$ for every $v\in X$, we get that
    \[|X|\leq \card*{\binom{|R|}{\leq s}} = \binom{d}{\leq s}.
    \qedhere
    \]
\end{proof}

We also need the following lemma to show that, with high probability,
most columns in an arbitrary small tuple of columns of $M_n$ have a small
fraction of their support intersecting the union of the support of their
preceding columns.

\begin{definition}
    \label{def:k-contained}
    For $\tau\defeq (i_1,\ldots, i_t)$ a tuple of distinct elements of $[m]$,
    we say that $i_j$ is \emph{$c$-contained w.r.t. $\tau$} if
    the set $\supp(M_n[i_j])$ has at least $c$ elements in common with $\bigcup_{p<j} \supp(M_n[i_p])$.
\end{definition}
\begin{lemma}
\label{M-prop2-R}
    Let $S_\tau$ be the following (random) set 
    \[
    S_\tau\defeq \setbar{j\in [t]}{\text{$i_j$ is $10k$-contained w.r.t. $\tau$}}.
    \]
    Then, for any positive integer $t\leq n^{0.1}$, we have
    \label{lem:M-prop2-R}
    \[
    \prob{M_n}{\text{$\exists\ \tau$ such that $|S_\tau| \geq t/2k$}}
    \leq \frac{1}{n}.
    \]
\end{lemma}
\begin{proof}
    Let $r\defeq \ceiling{t/2k}$ and $\kappa\defeq 10k$.
    By the union bound over the choices of $\tau\in [m]^t$ and $S\in\binom{[t]}{r}$, we get that
    \begin{align*}
        \prob{M_n}{\text{$\exists\ \tau$ such that $|S_\tau| \geq t/2k$}}
        \leq m^t\cdot \binom{t}{r}\max_{\tau, S}\prob{M_n}{S\subseteq S_\tau}.
    \end{align*}
    Now we are going to find an upper bound for $\prob{M_n}{S\subseteq S_\tau}$ for every $\tau\eqdef (i_1,\dotsc,i_t)$ and $S\in\binom{[t]}{r}$.
    Let $\set{s_1, \dotsc, s_r}\defeq S$ such that $s_1 < \dotsc < s_r$, and,
    for every $j\in [r]$, let
    \[T_j\defeq \cup_{p < s_j}\supp(M_n[i_p]).\]
    First note that $|T_j|\leq ts$, and
    \begin{align*}
       \prob{M_n}{S\subseteq S_\tau}
       &\leq \prob{M_n}{\cap_{j = 1}^r\set{|\supp(M_n[i_{s_j}])\cap T_j|\geq \kappa}} \\
       &= \prod_{j = 1}^r
       \prob{M_n}{|\supp(M_n[i_{s_j}])\cap T_j|\geq \kappa\bigg|
                  \cap_{i = 1}^{j-1}\set{|\supp(M_n[i_{s_i}])\cap T_i|\geq \kappa}}  \\
       &= \prod_{j = 1}^r
       \prob{M_n}{|\supp(M_n[i_{s_j}])\cap T_j|\geq \kappa},
    \end{align*}
    where the last equality follows from the independence between choices for the columns of $M_n$.
    For every $j\in [r]$, we have, again by independence of columns,
    \begin{align*}
        \prob{M_n}{|\supp(M_n[i_{s_j}])\cap T_j|\geq \kappa}
        &= \sum_{T\subseteq [n], \kappa\leq|T|\leq ts}
        \prob{M_n}{|\supp(M_n[i_{s_j}])\cap T|\geq \kappa}
        \prob{M_n}{T_j = T} \\
        &\leq \max_{T\subseteq [n], \kappa\leq|T|\leq ts}
        \prob{M_n}{|\supp(M_n[i_{s_j}])\cap T|\geq \kappa}.
    \end{align*}
    For every $T\subseteq [n]$ such that $\kappa\leq|T|\leq ts$,
    we get
    \begin{align*}
        \prob{M_n}{|\supp(M_n[i_{s_j}])\cap T|\geq \kappa}
        &= \prob{M_n}{\exists R\in\binom{T}{\kappa}\sttext R\subseteq \supp(M_n[i_{s_j}])} \\
        &\leq \frac{\binom{|T|}{\kappa}\binom{n-\kappa}{\leq (s-\kappa)}}{\binom{n}{\leq s}}\leq  \frac{\binom{|T|}{\kappa}\cdot (s+1)\cdot \binom{n-\kappa}{s-\kappa}}{\binom{n}{s}}\\
        &\leq (s+1)(tse/\kappa)^\kappa(2s/n)^\kappa \\
        &\leq n^{-0.5\kappa} = n^{-5k}. \\
    \end{align*}
    Therefore, we obtain
    \begin{align*}
       \prob{M_n}{S\subseteq S_\tau}
       &\leq n^{-5kr} \leq n^{-2.5t}
    \end{align*}
    and, as a consequence,
    \begin{equation*}
        \prob{M_n}{\text{$\exists\ \tau$ such that $|S_\tau| \geq t/2k$}}
        \leq \binom{t}{r}\cdot m^tn^{-2.5t}
        = \binom{t}{r}\cdot n^{-0.5t} \leq 1/n.
        \qedhere
    \end{equation*}
\end{proof}

We now collect in the following definition the properties of a matrix that we need for the proof of our lower bound.
\begin{definition}
    \label{defWellBehavedMatrix}
    We say that a matrix $M\in\Reals^{n\times m}$ is \emph{well-behaved} if the following properties hold:
    \begin{enumerate}
        \item Every column $M[i]$ has support size at least $s/2,$
        \item $\prob{S\in \binom{[m]}{10n\log n}}{\text{$M[S]$ is full rank}} \geq 1/10,$ and
        \item for any $t\leq n^{0.1}$ and any tuple $\tau= (i_1,\ldots, i_t)$ of distinct elements from $[m]$, the number of $j\in [t]$ such that $i_j$ is $10k$-contained w.r.t. $\tau$ is smaller than $t/2k.$
            In particular, if $t \leq k$, there are no $j\in [t]$ such that $i_j$ is $10k$-contained w.r.t. $\tau$.
\end{enumerate}
\end{definition}
By Lemma~\ref{lem:M-suppsize}, Lemma~\ref{lem:M-prop1-R} and Lemma~\ref{lem:M-prop2-R} (along with a union bound over all $t\leq n^{0.1}$), we know that, with positive probability, $M_n$ is a well-behaved matrix for every sufficiently large $n$.
For the rest of this section, we will denote by $M$ a fixed well-behaved matrix.

\subsection{Monotone Boolean circuit lower bounds via well-behaved matrices}

In the remaining of this section, we show that the Boolean function $f_M$
cannot be computed by a monotone Boolean circuit of small size.
This is the lower bound part of Theorem~\ref{thm:main3}.
As in Section~\ref{sectLowerBoundF2},
we start by defining two probability distributions over the inputs of $f_M$.

\begin{definition}
    \label{defDistributionAndError2}
    We define the distribution $a\sim D_1$ by sampling a uniformly random
    $a\in \{0,1\}^m$ of Hamming weight $10n\log n$. 
    We define the distribution $a\sim D_0$ by sampling a uniformly random $u\in \{-1,0,1\}^n$ and, for every $j\in [m]$, we set
    \begin{align*}
        a_j &\defeq 1 \text{ if $\iprod{M[j]}{u} = 0$, and} \\
        a_j &\defeq 0 \text{ otherwise,}
    \end{align*}
    where the inner product is taken over $\Reals$.
\end{definition}

\begin{observation}
    \label{obs:D0-D1-R}
    Since the matrix $M$ is well-behaved (see Definition~\ref{defWellBehavedMatrix}), the distribution $D_1$ satisfies $\prob{a\sim D_1}{f(a) = 1} \geq 1/10.$ 
    
    Further, as long as $u \neq 0$, the point $a$ sampled from $D_0$ is an
    element of $f_M^{-1}(0)$. Thus we have
    \[
    \prob{a\sim D_0}{f_M(a) = 0} = 1- \frac{1}{3^n}.
    \]
\end{observation}

Thus, to show a lower bound via the monotone circuit lower bound criterion
(\Cref{thm:sunflower-gen}),
as before
it suffices to show that $D_1$ is spread and to show
a bound for $(D_0,\eps)$-sunflowers.


\begin{lemma}[Spreadness of $D_1$]
    \label{lem:D1sparse-R}
    The distribution $D_1$
    is $n^{0.1}$-wise
    $(n/(10 \log n))$-spread.
\end{lemma}
\begin{proof}
    Let $T$ be a subset of $[m]$ of size $k\leq n^{0.1}$. 
    The proof of this lemma is similar to the proof of Lemma \ref{lem:D1sparse-F2}. 
    Note that, for $\ell \defeq 10n\log n$ and $S\sim \Unif(\binom{m}{\ell})$, 
    \begin{align*}
        \prob{a\sim D_1}{\bigwedge_{i\in T}a_i = 1}
        = \prob{S}{T\subseteq S}
        = \frac{\binom{m-k}{\ell-k}}{\binom{m}{\ell}}
        = \prod_{i=0}^{k-1} \frac{\ell-i}{m-i}
        \leq \left(\frac{l}{m}\right)^k,
    \end{align*}
    where the last inequality follows from $\frac{\ell-i}{m-i}\leq
    \frac{\ell}{m}$ for $i\leq k-1$.
    As $m=n^2$, we have
    \begin{equation*}
        \prob{S}{T\subseteq S}
        \leq \left(\frac{10n\log n}{n^2}\right)^k 
        =
        (n/(10 \log n))^{-k}.
        \qedhere
    \end{equation*}
\end{proof}

\paragraph{Sunflower bound.}

To show the sunflower bound, we first show a weak form of
bounded independence for $D_0$.

\begin{lemma}[Weak form of independence for $D_0$]
\label{lem:D0ind-Reals}
Assume that $t\leq n^{0.1}$. Fix any tuple $\tau \defeq (i_1,\ldots, i_t)$ of distinct elements in $[m]$ and $j\in [t]$ such that $i_j$ is not $10k$-contained w.r.t. $\tau$. Then, 
\[
    \prob{a\sim D_0}{a_{i_j} = 1\ |\ a_{i_1},\ldots,a_{i_{j-1}}} \geq \Omega(1/k).
\]
\end{lemma}
\begin{proof}
    Let $S_\ell\defeq\supp M[i_\ell]$ for $\ell\in[j]$, and 
    $R\defeq S_j\setminus \bigcup_{\ell=1}^{j-1} S_\ell$.
    For every $T\subseteq [n]$, let 
    \[X_T\defeq \sum_{l\in T} u_l,\]
    where $u_l$'s are the random variables used to defined $D_0$.
    If $i_j$ is not $10k$-contained w.r.t. $\tau$, we have
    \[|R|\geq |S_j| - 10k\geq (s/2)-10k.\]
    As $X_{[n]} = 0$, we have that
    \begin{align*}
        \Pr_{a\sim D_0}&[ a_{i_j} = 1\ |\ a_{i_1},\ldots,a_{i_{j-1}}]
        \\
        &= \sum_{p = -10k}^{10k} \prob{a\sim D_0}{X_{S_j\drop R} = -p, X_R = p \;|\;a_{i_1},\ldots,a_{i_{j-1}}} \\
        &= \sum_{p = -10k}^{10k} \prob{a\sim D_0}{X_R = p \;|\;X_{S_j\drop R} = -p, a_{i_1},\ldots,a_{i_{j-1}}}
                                 \prob{a\sim D_0}{X_{S_j\drop R} = -p\;|\;a_{i_1},\ldots,a_{i_{j-1}}} \\
        &= \sum_{p = -10k}^{10k} \prob{a\sim D_0}{X_R = p}
                                 \prob{a\sim D_0}{X_{S_j\drop R} = -p\;|\;a_{i_1},\ldots,a_{i_{j-1}}},
    \end{align*}
    where the last equality follows from the independence between coordinates of $u$. 
    Note that, for every $p\in [-10k, 10k]$ and $r\defeq |R|$ and
    $Y_l\sim \mathrm{Rademacher}(l)$ being a Rademacher random variable with $l$ independent samples for every $l\in\set{0,\dotsc, r}$,
    we have
    \begin{align*}
        \prob{a\sim D_0}{X_R = p}
        &= \sum_{U\subseteq R} \prob{a\sim D_0}{\inv{u}(0) = U, X_R = p} \\
        &= \sum_{U\subseteq R} \prob{a\sim D_0}{X_{R\drop U} = p\;|\;\inv{u}(0) = U}
                              \prob{a\sim D_0}{\inv{u}(0) = U} \\
        &\geq \sum_{U\subseteq R, ||U| - r/3|\leq \delta r/3} \prob{a\sim D_0}{Y_{r-|U|} = p}
                              \prob{a\sim D_0}{\inv{u}(0) = U} \\
        &\geq \left(\min_{w\in [r]\sttext |w - r/3|\leq \delta r/3} \prob{a\sim D_0}{Y_{r-w} = p}\right)
              \left(\sum_{U\subseteq R, ||U| - r/3|\leq \delta r/3} \prob{a\sim D_0}{\inv{u}(0) = U}\right),
    \end{align*}
    for any $\delta \in [0, 1]$.
    Note that $\inv{u}(0)\sim\mathrm{Bin}(r, 1/3)\eqdef B$;
    thus, by Chernoff's inequality~\cite{MitzUpfalPC17},
    \[\prob{B}{|B - \Exp B| > \delta\Exp B} \leq 2e^{-\delta^2\Exp B/3},\]
    and, for $\delta\defeq \sqrt{18/r}$,
    \[\prob{B}{|B - r/3| > \delta r/3} \leq 2e^{-2}\leq 1/2.\]
    Hence,
    \[\sum_{U\subseteq R, ||U| - r/3|\leq \delta r/3} \prob{a\sim D_0}{\inv{u}(0) = U} \geq 1/2.\]
    Now we just need to find good estimates for $\prob{a\sim D_0}{Y_v = p}$ for $v\defeq r - w$ with $w\in [r]$ satisfying
    \[|w - r/3| \leq \sqrt{2r}.\]
    In order to obtain these estimates, we will use standard inequalities to deal with binomial coefficients.
    First note that
    \begin{align*}
        \prob{a\sim D_0}{Y_v = p} = \binom{v}{(v + p)/2}2^{-v}.
    \end{align*}
    By using that $1-x\geq e^{-x/(1-x)}$ for every $x\leq 1$, we get that
    \begin{align*}
        \frac{\binom{v}{(v+p)/2}}{\binom{v}{v/2}}
        &= \prod_{i=0}^{p/2-1} \frac{v/2 - i}{v/2 + p/2 - i} \\
        &\geq \left(\frac{v/2 - p/2}{v/2}\right)^{p/2} \\
        &\geq \left(e^{-\frac{p}{v}\frac{1}{1-p/v}}\right)^{p/2} \\
        &\geq e^{-\frac{p^2}{2(v-p)}},
    \end{align*}
    and, using that $\binom{v}{v/2} = \Theta(2^v/\sqrt{v})$, we obtain
    \begin{align*}
        \prob{a\sim D_0}{Y_v = p}
        &\geq e^{-\frac{p^2}{2(v-p)}}\binom{v}{v/2}2^{-v} \\
        &\geq \Omega\left(\frac{e^{-\frac{p^2}{2(v-p)}}}{\sqrt{v}}\right).
    \end{align*}
    As $p\in [-10k, 10k]$ and
    \[200k^2 = s\geq v
    \geq 2r/3 - \sqrt{18r}
    \geq r/3
    \geq \frac{1}{3}((s/2) - 10k)
    = \frac{1}{3}(200k^2 - 10k)\geq 60k^2\]
    for sufficiently large $k$, we get
    \begin{align*}
        \prob{a\sim D_0}{Y_v = p}
        &\geq \Omega\left(\frac{e^{-\frac{200k^2}{2\cdot (60k^2 - 10k)}}}{\sqrt{k^2}}\right) \\
        &\geq \Omega\left(1/k\right).
    \end{align*}
    Therefore, we obtain
    \begin{equation*}
        \prob{a\sim D_0}{a_{i_j} = 1\ |\ a_{i_1},\ldots,a_{i_{j-1}}}
        \geq \Omega\left(1/k\right)\cdot \frac{1}{2}
        = \Omega(1/k).
        \qedhere
    \end{equation*}
\end{proof}

Our main combinatorial tool is the classical sunflower lemma~\cite{ER60} (see also~\cite[Section 6.1]{Jukna11}).
Recall that a sunflower is a collection of sets $S_1,\ldots, S_r$ such that
the pairwise intersections $S_i\cap S_j$ are all the same. 
The improved bounds of~\cite{ALWZ21} and subsequent works will not make any
substantial difference in our bounds.

\begin{lemma}[Sunflower lemma]
\label{lem:SL}
    If $\cals$ is a family of sets of size at most $\ell \in [k]$
    such that
    $\card{\cals} \geq \ell!(r-1)^\ell$,
    then $\cals$
    contains a sunflower with $r$ sets.
\end{lemma}

We now prove our $D_0$-sunflower bound.
To make the bound cleaner and simpler to prove, we now
set our choice of $k$.
We set
\begin{equation}
    \label{eq:k-bound}
    k \defeq \ceiling{(\log m)^{1/2}}.
\end{equation}
This is the only place where we need to set $k = n^{o(1)}$,
owing to our use of the \emph{classical} sunflower lemma.

\begin{lemma}[$D_0$-sunflower lemma]
    \label{lem:R-neg-sunflower}
    For every $\ell \in [k]$, we have
    $r(D_0, \ell, m^{-4k}) \leq 2k^{2\ell+1} \ell \log m$.
\end{lemma}
\begin{proof}
    Let $\eps \defeq m^{-4k}$.
    Let $\cals$
    be a $\ell$-uniform family of sets
    larger than
    $(2 k^{2\ell} \ell \log(1/\eps))^\ell$.
    By \Cref{lem:SL},
    there exists a sunflower $\cals' \sseq \cals$ with
    \[r \defeq 2k^{2\ell} \log(1/\eps)\]
    sets with some core
    $K \defeq \bigcap_{S\in\cals'} S$.
    Note that
    \begin{equation*}
        r
        = 
        8 k^{2\ell + 1}\log m
        \ll
        n^{0.1},
    \end{equation*}
    by our choice of $k = \sqrt{\log m}$ and $m = n^2$
    (\Cref{eq:params-reals}).

    We now show that $\cals'$ is a $(\cald_0,\ell,\eps)$-sunflower.
    Let $\calf \defeq \set{S \sm K : S \in \cals'}$.
    It suffices to show that, for $F \defeq \indfml{\calf}$,
    we have
    \begin{equation*}
        \Pr_{a \flws \cald_0}
        [F(a) = 0] \leq \eps.
    \end{equation*}
    Let $F_1, F_2, \ldots, F_{r}$ be the 
    sets of $\calf$.
    Let 
    \[\tau\defeq (i_1, i_2, \ldots, i_t)\]
    be the sequence of all the indices corresponding to the variables appearing in the terms
    $F_1, F_2, \ldots, F_{r}$ in that order (inside each term we order the variables arbitrarily).
    Since $t \leq rk \leq n^{0.1}$ and by the choice of $M$ 
    (see \Cref{defWellBehavedMatrix}, Item (3)),
    there are at most $t/2k\leq r/2$ indices in $\tau$ which are $10k$-contained with respect to $\tau$. 
    In particular, these indices appear across at most $r/2$ `corrupted' terms among $F_1, \ldots, F_{r}$.
    In the remaining terms, all of their indices are not $10k$-contained w.r.t. $\tau$.
    Removing the corrupted terms yields a 
    a subfamily $\calf' \sseq \calf$
    with  $r'\geq r/2$
    sets
    such that all their indices are not $10k$-contained w.r.t.
    $\tau$.
    Let $F_1',\ldots, F'_{r'}$ be the 
    sets of $\calf'$.

    For any $j\in [r']$, we can apply Lemma \ref{lem:D0ind-Reals} for each
    of $F'_j$'s elements (recall that $F'_j$ has at most $\ell$ literals),
    and then obtain
    \[
        \prob{a\sim D_0}{F'_j(a)=1 \mid F'_1(a)=0, \ldots, F'_{j-1}(a)=0}
        \geq \Omega(1/k)^\ell. 
    \]
    Therefore, we obtain
    \[\prob{a\sim D_0}{F'_j(a)=0 \mid F'_1(a)=0, \ldots, F'_{j-1}(a)=0}\leq 1-\Omega(1/k)^\ell,\]
    and, as a consequence, 
    \[
        \prob{a\sim D_0}{F(a)=0}
        \leq \prob{a\sim D_0}{F'(a)=0}\leq \left(1-\Omega(1/k)^\ell \right)^{r'} 
        \leq 
        \exp(-rk^{-2\ell}/2)
        = \eps,
    \]
    for sufficiently large $k$.
\end{proof}

\paragraph{Wrapping up.}
We can now apply \Cref{thm:sunflower-gen},
finishing the proof.

\begin{proof}[Proof of \Cref{thm:main3}]
    We have
    shown that there is a sequence of matrices $(M_n)_{n\in\Naturals}$ 
    with entries from $\mathbb{R}$ such that
    $M_n$ is an $n\times m$ matrix where $m \defeq n^2$
    and $M_n$ is \emph{well-behaved} (\Cref{defWellBehavedMatrix}).
    We have also exhibited two distributions $D_1, D_0$
    (\Cref{defDistributionAndError2})
    supported over strings
    $\blt^m$
    such that
    \begin{enumerate}
        \item 
            $\Pr_{x \flws D_i}[f_M(x)=i] 
            \geq \alpha$
            for every $i \in \blt$, where $\alpha > 0$ is some constant
            (\Cref{defWellBehavedMatrix} and \Cref{obs:D0-D1-R});
        \item $D_1$ is $n^{0.1}$-wise $(n/(10 \log n))$-spread
            (\Cref{lem:D1sparse-R});
        \item $r(D_0, \ell, m^{-4k}) \leq 8 k^{2\ell+1} \log m$ for every
            $\ell \in [k]$
            (\Cref{lem:R-neg-sunflower}).
    \end{enumerate}
    Taking $w\defeq k$
    in
    \Cref{thm:sunflower-gen} and noting that
    $\alpha m^{-3k} \leq m^{-4k}$, we obtain that
    there exists a constant $\beta > 0$
    such that the monotone complexity of $f_M$
    is
    \begin{equation*}
        \left( 
            \frac{\beta n}{k^{2k+1} \log^2 n }
        \right)^{k}
        =
        m^{\Omega(\sqrt{\log m})}.
    \end{equation*}
    The nonmonotone circuit upper bound for $f_M$ was proved in \Cref{lemmaAlgebraicUBReals}.
\end{proof}

\small
\bibliographystyle{alphaurl}
\bibliography{refs.bib} 
\appendix
 
\end{document}